\documentclass[12pt]{article}
\usepackage{lMac}
\usepackage{becMac}
\renewcommand{\theequation}{\arabic{equation}}

\newcommand{\Rho}{{R}}
\newcommand{\Chi}{{X}}
\newcommand{\Eta}{{H}}
\newcommand{\tD}{{\tilde \bbbd}}

\def\showintremarks{n}   
\def\showissues{n}       

\begin{document}
\title{Complex Bosonic Many--body Models:\\
        \Large Overview of the Small Field Parabolic Flow}

\author{Tadeusz Balaban}
\affil{\small Department of Mathematics \authorcr
       Rutgers, The State University of New Jersey \authorcr
       tbalaban@math.rutgers.edu\authorcr
       \  }

\author{Joel Feldman\thanks{Research supported in part by the Natural 
                Sciences and Engineering Research Council 
                of Canada and the Forschungsinstitut f\"ur 
                Mathematik, ETH Z\"urich.}}
\affil{Department of Mathematics \authorcr
       University of British Columbia \authorcr
       feldman@math.ubc.ca \authorcr
       http:/\hskip-3pt/www.math.ubc.ca/\squig feldman/\authorcr
       \  }

\author{Horst Kn\"orrer}
\author{Eugene Trubowitz}
\affil{Mathematik \authorcr
       ETH-Z\"urich \authorcr
       knoerrer@math.ethz.ch, trub@math.ethz.ch \authorcr
       http:/\hskip-3pt/www.math.ethz.ch/\squig knoerrer/}


\maketitle

\begin{abstract}
\noindent
This paper is a contribution to a program to see symmetry breaking in a
weakly interacting many Boson system on a three dimensional lattice at 
low temperature.  It provides an overview of the analysis, given in 
\cite{PAR1,PAR2}, of the ``small field''  approximation to the ``parabolic flow'' 
which exhibits the formation of a ``Mexican hat'' potential well. 
\end{abstract}


\newpage

It is our long term goal to rigorously demonstrate symmetry breaking in a 
gas of bosons hopping on a three dimensional lattice. 
Technically, to show that the correlation functions decay at a non--integrable 
rate when the chemical potential is sufficiently positive, 
the non--integrability reflecting the presence of a long range Goldstone boson mediating the interaction between quasiparticles in the superfluid condensate.  
It is already known \cite{BryFed1,BryFed2} that the correlation functions 
are exponentially decreasing when the chemical potential is sufficiently 
negative. See, for example, \cite{Col} and \cite[\S19]{Wein} for an introduction
to symmetry breaking in general, and \cite{AGD,BOG,FW,PS} as 
general references to Bose-Einstein condensation. See \cite{Ben,CG,LSSY,Sei} 
for other mathematically rigorous work on the subject.

We start with a brief, formula free, summary of the program
and its current state. Then we'll provide a more precise,
but still simplified, discussion of the portion of the program
that controls the small field parabolic flow. 

\medskip

The program was initiated in \cite{fnlint1,fnlint2}, where we expressed the 
positive temperature partition function and thermodynamic 
correlation functions in a periodic box (a discrete three--dimensional
torus) as `temporal' ultraviolet limits of four--dimensional (coherent
state) lattice functional integrals (see also \cite{NO}). By a lattice 
functional integral we mean an integral 
with one (in this case complex) integration variable for each point of 
the lattice. By a `temporal' ultraviolet limit, we mean a limit in which 
the lattice spacing in the inverse temperature direction (imaginary time 
direction) is sent to zero while the lattice spacing in the three 
spatial directions is held fixed.

In \cite{UV}\footnote{See also \cite{LH} for a more pedagogical introduction.}, 
by a complete large field/small field renormalization group 
analysis, we expressed the temporal ultraviolet limit for the partition 
function\footnote{A similar analysis will yield the corresponding
representations for the correlation functions.}, still in a periodic
box, as a four--dimensional 
lattice functional integral with the lattice spacing in all four directions
being of the order one, preparing the way for an infrared renormalization 
group analysis of the thermodynamic limit.

This overview concerns the next stage of the program, which is
contained in \cite{PAR1,PAR2} and the supporting papers
\cite{PARL,Bloch,POA,BlockSpin,BGE,SUB}. 
There we initiate the infrared analysis by tracking, in the small field
region, the evolution of the 
effective interaction generated by the iteration of a renormalization group map 
that is taylored to a parabolic covariance\footnote{Morally, 
the $1+3$ dimensional heat operator.}: in each renormalization group step 
the spatial lattice directions expand by a factor\footnote{$L$
is a fixed, sufficiently large, odd natural number.} $L>1$, 
the inverse temperature direction expands by a factor $L^2$ and
the running chemical potential grows by a factor of $L^2$,
while the running coupling constant decreases by a factor of $L^{-1}$.
Consequently, the effective potential, initially close to a paraboloid, 
develops into a Mexican hat with a moderately large radius and a 
moderately deep circular well of minima. \cite{PAR1,PAR2} ends after a finite 
number (of the order of the magnitude of the logarithm of the coupling constant)
of steps once the chemical potential, which initially was of the
order of the coupling constant, has grown to a small `$\epsilon$'
power of the coupling constant. Then we can no longer base our
analysis on expansions about zero field, because the renormalization 
group iterations have moved the effective model away from the 
trivial noninteracting fixed point.

In the next stage of the construction, we plan to continue the parabolic 
evolution in the small field regime, but expanding around fields 
concentrated at the bottom of the 
(Mexican hat shaped) potential well rather around zero (much as is done 
in the Bogoliubov Ansatz) and track it through an additional finite 
number of steps until the running chemical potential is sufficiently 
larger than one. 
At that point we will turn to a renormalization 
group map with a scaling taylored to an elliptic covariance,  
that expands both the  temporal (inverse temperature)
and spatial lattice directions by the same factor $L$. It is expected that 
the  elliptic evolution can be controlled through infinitely many steps, 
all the way to the symmetry broken fixed point. The system is superrenormalizable in the entire parabolic regime because the running 
coupling constant is geometrically decreasing. However in the elliptic 
regime, the system is only strictly renormalizable. 

The final stage(s) of the program concern the control of the
large field contributions in both the parabolic and elliptic
regimes.

\medskip

The technical implementation of the parabolic renormalization group 
in \cite{PAR1,PAR2} proceeds much as in \cite{CPS,UV}, except
that we are restricting our attention to the small field regime
and 
\begin{itemize}[leftmargin=0.5in, topsep=2pt, itemsep=0pt, parsep=0pt]
\item[$\circ$]
we use $1+3$ dimensional block spin averages, as in \cite{KAD,Bal1,GK}. 
In \cite{UV}, we had used decimation, which was suited to 
the effectively one dimensional problem of evaluating the temporal ultraviolet 
limit. 
\item[$\circ$]
Otherwise, the stationary phase calculation that controls oscillations 
is similar, but technically more elaborate. 
\item[$\circ$]
The essential complication is that the critical fields and background 
fields are now solutions to (weakly) nonlinear systems of parabolic equations. 
\item[$\circ$]
The Stokes' argument that allows us to shift the multi dimensional 
integration contour to the `reals' and 
\item[$\circ$]
the evaluation of the fluctuation integrals is similar. 
\item[$\circ$]
However, there is an important new feature: the chemical potential 
has to be renormalized.
\end{itemize}

To analyze the output of the block spin convolution (a single 
renormalization group step), it is de rigueur for the small field/large
field style of renormalization group implementations to introduce local small 
field conditions on the integrand and then decompose the integral into 
the sum over all partitions of the discrete torus 
into small and large field regions on which the conditions are satisfied
and violated, respectively. 
Small field contributions are to be controlled by powers of the 
coupling constant $\, \fv_0\, $ (a suitable norm of the two body  interaction) uniformly in the volume of the small field region. 
Large field contributions are to be controlled by a factor 
$\,e^{-{1/\fv_0^\veps}}\ ,\ \veps\, >\, 0\, $,
raised to the volume of the large field region. Morally, in small field
regions, perturbation expansions in the coupling constant converge and 
exhibit all physical phenomena. Large field regions give multiplicative 
corrections that are smaller than any power of the coupling constant. So, 
in the leading terms, every point is small field. 

If the actions in our functional integrals were  sums of positive terms 
(as in a Euclidean O(n) model) it would be routine to extract an
exponentially small factor per point of a large field region. 
They are not. There are explicit purely imaginary terms.
In \cite{PAR1,PAR2} we analyze the parabolic flow of the leading term,
in which all points are small field, as long as it is possible to 
expand around zero field. Nevertheless, we show (see, 
\cite{PARL}) that our actions \emph{do} have positivity
properties and consequently there is at least one factor 
$\,e^{-1/\fv_0^\veps}\, $ whenever there is a large field region. 
A stronger bound of a factor per point of a large field region is 
reasonable and would be the main ingredient for controlling the 
full parabolic renormalization group flow in this regime.

\medskip

We now formally introduce the main objects of discussion
and enough machinery to allow technical (but simplified) statements of
the main results of \cite{PAR1,PAR2} and the methods used to establish them.

One conclusion of our previous work in \cite{UV} is that the purely 
small field contribution to the partition function for 
a gas of bosons hopping on a three dimensional discrete torus
$\,X = \bbbz^3/L_\sp\bbbz^3\, $ (where $L_\sp$, a power of $L$,
is the spatial infrared regulator which will ultimately be sent to
infinity) takes the form 
\begin{equation}\label{eqnOVinitialFnlInt}
\int_{S_0}
\prod_{x\, \in\, \cX_0}\!\!
            \sfrac{d\psi(x)^*\wedge d\psi(x)}{2\pi\imath}\
e^{\cA_0(\psi^*,\, \psi)}
\end{equation}
where
\begin{itemize}[leftmargin=*, topsep=2pt, itemsep=0pt, parsep=0pt] 
\item 
$\cX_0\, =\, \bbbz/L_\tp\bbbz\, \times\, X\,$ is a $1+3$ 
dimensional discrete torus
with points $\, x\, =\, (x_0,\, \bx)\, $. Here, 
$L_\tp \approx \frac{1}{kT}$, also a power of $L$, is the inverse temperature infrared 
regulator, which can ultimately be sent to infinity to get the
temperature zero limit. 

\item
$\,\psi\, \in\, \bbbc^{\cX_0}\, $
is a complex valued field on $\,\cX_0\, $,
 $\, \psi^*\, $ is the complex conjugate field
and, for each $x\in\cX_0$,
$\,  
  \frac{d\psi(x)^*\wedge d\psi(x)}{2\imath}
\, $ 
is the standard Lebesgue measure on $\bbbc$.

\item
$
S_0
=  
\big\{\, \psi\, \in\,  \bbbc^{\cX_0}
\  \big|\  
|\psi(x)|\, \le  \fv_0^{-\nicefrac{1}{3}+\eps},\ 
|\partial_\nu\psi(x)|\, \le 
 \fv_0^{-\nicefrac{1}{3}+\eps}
\, ,\  \nu\, =\, 0,1,2,3\ ,\ x\, \in\, \cX_0 
\, \big\}
\, $,
where  the small `coupling constant' $\,\fv_0\, $ is an 
exponentially,  tree length weighted $L^1$--$L^\infty$--norm 
(see the discussion  of norms at the end of this overview or 
\cite[Definition \defHTkernelnorm]{PAR1}) of an effective interaction 
$\,V_0\,$ (see \cite[Proposition \propSZprepforblockspin]{PAR1}).
Here, 
$\, \partial_\nu\, ,\, \nu\, =\, 0,\, 1,\, 2,\, 3\, $, is the forward difference operator 
in the $\,x_\nu\, $ direction.

\item
Let $\,\psi_*\, $
be another arbitrary element of $\,\bbbc^{\cX_0}\, $.
($\,\psi_*\, $ is {\it not} to be confused with the complex conjugate 
$\, \psi^*\, $ of $\,\psi\, $.)

\item
$\,\cA_0(\psi_*,\, \psi) = 
    -A_0(\psi_*,\, \psi) 
    +p_0(\psi_*,\, \psi,\, \nabla\psi_*,\, \nabla\psi)
\, $.
The action $\, \cA_0(\psi^*,\, \psi)\, $ determining the partition 
function is the restriction $\,\cA_0(\psi^*,\, \psi)\, 
=\, \cA_0(\psi_*,\, \psi)\,  \big|_{\psi_*=\psi^*}
\, $ of $\, \cA_0(\psi_*,\, \psi)\, $ to the `real' subspace
$\,\psi_*=\psi^*\, $ of 
$\, \bbbc^{\cX_0}\times \bbbc^{\cX_0}\, $.
Here,  $\, \nabla\, $ is the (four dimensional)  discrete gradient operator.

\item `Morally', $\, A_0(\psi_*,\, \psi)
\, =\, 
\<\psi_*,\, (-\partial_0+h)\psi\>_0
+\cV_0(\psi_*,\, \psi)
- \mu_0 \<\psi_*,\, \psi\>_0
\, $
,
where 
   \begin{itemize}[leftmargin=*, topsep=2pt, itemsep=0pt, parsep=0pt] 
       \item
           $\,\<f\, ,\, g\>_0 = 
           \sum\limits_{x\, \in\,\cX_0 }\!f(x)g(x)\, $
           is the natural real inner product on $\,\bbbc^{\cX_0}\, $
      \item 
           $\, h\, $ is a nonnegative, 
           second order, elliptic (lattice) pseudodifferential operator acting 
           on  $\, X \, $ --- for example, a constant times minus
           the spatial discrete laplacian
      \item
           $\,\cV_0(\psi_*,\, \psi) =
            \frac{1}{2}\, \sum\limits_{ \cX_0^4}
            V_0(x_1,x_2,x_3,x_4)\,
            \psi_*(x_1) \psi(x_2)
            \psi_*(x_3) \psi(x_4)
            \, $
            is a quartic monomial whose kernel $\, V_0\, $
            is translation invariant with
            $\,\sum\limits_{ \cX_0^3 }\, V_0(0,x_2,x_3,x_4) > 0\,$
      \item
            $\, \mu_0\, $ is essentially the chemical potential.
   \end{itemize}

\item 
Let $\, \psi_{*\nu},\, \psi_\nu\, ,\, 
\nu\, =\, 0,\,1,\,2,\,3\, $, be the names of new arbitrary elements of 
$\, \bbbc^{\cX_0}\, $. The perturbative correction 
$\, p_0\big(\psi_*,\psi,\, \{\psi_{*\nu}\}_{\nu=0}^3,\, 
\{\psi_\nu\}_{\nu=0}^3 \big), $ to the principal contribution 
$\,-A_0\, $, in $\,\cA_0$, is a power series in the ten variables
$\,\psi_*,\, \psi,\, \{\psi_{*\nu},\, \psi_\nu\}_{\nu=0}^3$, 
with no $\,\psi_*(x)\psi(y)\, $ terms, such that each nonzero term 
has as many factors with asterisks as factors without asterisks.
That is, $p_0$ conserves particle number.  It converges on 
\begin{equation*}  
\Big\{\big(\psi_*, \psi, 
          \{\psi_{*\nu}, \psi_\nu\}_{\nu=0}^3\big) \in  \bbbc^{10\cX_0}
\  \Big|\  
|\psi_{(*)}(x)|,|\psi_{(*)\nu}(x)| \le  \fv_0^{-\nicefrac{1}{3}+\veps}
,\  0\le\nu\le 3,\,x\in\cX_0 
\, \big\}
\end{equation*}
where ``$(*)$'' means ``either with $*$ or without $*$''.

\end{itemize}
See \cite[Proposition \propSZprepforblockspin]{PAR1} for more details.
\bigskip

For convenience, set
\begin{equation*}
F_0(\psi^*,\psi)
\, =\, 
e^{\cA_0(\psi^*,\, \psi)}
\,
\chi_{S_0}(\psi)
\end{equation*}
With this notation the partition function is
\begin{equation}\label{eqnPOVpartFn}
\int\! \prod_{x\, \in\, \cX_0}\!\!
            \sfrac{d\psi(x)^*\wedge d\psi(x)}{2\pi\imath}
\ F_0(\psi^*,\psi)\ +\ O\big(e^{-1/\fv_0^\veps}\, \big)
\end{equation}

\bigskip

It is natural to study the partition function 
using a steepest descent/stationary phase analysis.
The exponential $\,e^{\<\psi^*,\,  \partial_0 \psi\>}\, $
is purely oscillatory because the quadratic form
$\,\<\psi^*,\,  \partial_0 \psi\>\, $
is pure imaginary. Fortunately, our partition function,
$\,\cZ\, $, has the essential feature that there is an   analytic function 
$\, \cA_0(\psi_*,\, \psi)\, $ on a neighborhood of the origin in
$\,\bbbc^{\cX_0}\times \bbbc^{\cX_0}\, $ whose restriction to the real 
subspace is the  `small field' action.
Our renormalization group analysis of the   
oscillating integral defining $\,\cZ\, $ is based on
the critical points of $\,A_0(\psi_*,\, \psi)
\, =\, 
\<\psi_*,\, (-\partial_0+h)\psi\>
\, +\, 
\cV_0(\psi_*,\, \psi)
\, -\,
\mu_0\, 
\<\psi_*,\, \psi\>
\, $
in $\,\bbbc^{\cX_0}\times \bbbc^{\cX_0}\, $ that  
typically do not lie in the real subspace, and a multi dimensional
Stokes' contour shifting construction that is only possible
because $\,p_0(\psi_*,\psi)\, $ is analytic.

\bigskip

We now formally introduce the `block spin' renormalization group 
transformations that are used in this paper. Let $\,\cX_{-1}\, $ be the subgroup 
$\, L^2\bbbz/L_\tp\bbbz\, \times\, L\bbbz^3/L_\sp\bbbz^3\, $ of $\,\cX_0\, $. 
Observe that the  distance between points of $\cX_{-1}$
on the inverse temperature axis is $\,L^2\, $
and on the spatial axes is $\,L\, $, and that
$\,|\cX_{-1}|\, =\, L^{-5}|\cX_0|\, $. Also, let 
$\, \cQ^{(0)}\, :\, \bbbc^{\cX_0}\, \rightarrow\, \bbbc^{\cX_{-1}}\, $ be a 
linear operator  that commutes with complex conjugation.
We will make a specific choice of $\cQ^{(0)}$ later. It will
be a `block spin averaging' operator with, for each $y\in\cX_{-1}$,
$\big(\cQ^{(0)}\psi\big)(y)$ being `morally' the average value
of $\psi$ in the $L^2\times L\times L\times L$ block centered on
$y$. Insert into the integral of \eqref{eqnPOVpartFn}
\begin{equation*}
1=\sfrac{1}{N^{(0)}}\int_{\bbbc^{\cX_{-1}}}\, \prod_{y\, \in\, \cX_{-1}}\!\!
            \sfrac{d\th(y)^*\wedge d\th(y)}{2\pi\imath}\ 
e^{-\frac{1}{L^2}\<\right.\theta_*\, -\, \cQ^{(0)}\psi^*\, ,\, 
         \theta \, -\, \cQ^{(0)}\psi\left.\>_{-1} }
\end{equation*}
where $\,\<f\, ,\, g\>_{-1}\, =\, L^5
\sum_{y\, \in\,\cX_{-1} }f(y)g(y)\, $
is the natural real inner product on $\,\bbbc^{\cX_{-1}}\, $
and $\,N^{(0)}\, $ is a normalization constant. Then exchange
the order of the $\psi$ and $\th$ integrals. This gives 
\begin{equation*}
\int\! \prod_{x\, \in\, \cX_0}\!\!
            \sfrac{d\psi(x)^*\wedge d\psi(x)}{2\pi\imath}
\ F_0(\psi^*,\psi)
=
\int \prod_{y\, \in\, \cX_{-1}}\!\!
            \sfrac{d\th(y)^*\wedge d\th(y)}{2\pi\imath}\ 
\ B_1(\theta_*,\, \theta)
\end{equation*}
where, by definition, the block spin transform of 
$\, \,F_0(\psi_*,\, \psi)\, $ associated to 
$\, \cQ^{(0)}\, $ with external fields $\,\theta\, $ and $\,\theta_*\,$  
is
\begin{equation*}
B_1(\theta_*,\, \theta)
=
\sfrac{1}{N^{(0)}}\int_{\bbbc^{\cX_0}}
   \prod_{x\, \in\, \cX_0}\!\!
            \sfrac{d\psi(x)^*\wedge d\psi(x)}{2\pi\imath}\  
e^{-\frac{1}{L^2}\<\right.\theta_*\, -\, \cQ^{(0)}\psi^*\, ,\, 
         \theta \, -\, \cQ^{(0)}\psi\left.\>_{-1} }
\
F_0(\psi^*,\, \psi)
\end{equation*}
Here $\, \theta\, ,\, \theta_*\, $ are
two arbitary elements  of $\,\bbbc^{\cX_{-1}}\, $.

It can be awkward to compare functions defined on discrete tori
with different lattice spacings. So, we scale $\cX_{-1}$ down
to the unit discrete torus  
\begin{align*}
\cX_0^{(1)}\, =\, \bbbz/\sfrac{L_\tp}{L^2}\bbbz\, \times\, \bbbz^3/\sfrac{L_\sp}{L}\bbbz^3
\end{align*} 
using the `parabolic' scaling map $\,x\, \in\, \cX_0^{(1)}\, 
\rightarrow\, (L^2x_0,\, L\bx)\, \in\, \cX_{-1}\, $,
which is an isomorphism of Abelian groups. 
Abusing notation, we  consciously use the symbol $\,\psi(x)\, $ 
as the name of a field on the unit torus $\,\cX_0^{(1)}\, $ even 
though it was used before as the name of a field on the unit torus 
$\,\cX_0\, $. By definition, the block spin renormalization group
transform of $\, F_0(\psi^*,\, \psi)\, $ associated to 
$\,\cQ^{(0)}\,$ with external fields $\,\psi\, $ and 
$\,\psi_*\, $ in $\,\bbbc^{\cX_0^{(1)}}\, $ is
\begin{equation}\label{eqnOVscalingMap}
F_1(\psi_*,\, \psi)
= B_1\big(\bbbs^{-1}\psi_*,\, \bbbs^{-1}\psi\big) 
\qquad
\text{where} 
\qquad
\big(\bbbs^{-1}\psi\big)(y_0,\by)\, =\, 
L^{-\nicefrac{3}{2}}\,\psi\big(\sfrac{y_0}{L^2},\sfrac{\by}{L}\big)
\end{equation} 
for any $\,\psi\, \in\, \bbbc^{\cX_0^{(1)}}$.
The `parabolic' exponent $\,-\nicefrac{3}{2}\, $ has been
chosen so that\footnote{In
         $\<\theta^*,(\partial_0+\Delta)\theta\>_{-1}$,
         $\partial_0$ is the forward difference operator on 
         $\cX_{-1}$. That is, $\, (\partial_0f)(y) = 
             \frac{f(y_0+L^2,\by)-f(y_0,\by)}{L^2}$.
         Similarly, for spatial difference operators.} 
$\, 
\<\bbbs \theta^*,\,  (\partial_0+\Delta)\bbbs \theta\>_0
=
\<\theta^*,\,  (\partial_0+\Delta)\theta\>_{-1}
\, $.
We now have
$$
L^{-3|\cX_0^{(1)}|}\, 
\int_{\lower3pt\hbox{${\sst\bbbc^{\cX_0^{(1)}}}$}}  
   \prod_{x\, \in\, \cX_0^{(1)}}\!\!
            \sfrac{d\psi(x)^*\wedge d\psi(x)}{2\pi\imath}\  
    F_1(\psi^*,\, \psi)
=
\int_{\bbbc^{\cX_0}}
   \prod_{x\, \in\, \cX_0}\!\!
            \sfrac{d\psi(x)^*\wedge d\psi(x)}{2\pi\imath}\  
    F_0(\psi^*,\, \psi)
$$
the original small field part of the partition function.

Repeat the construction. 
\begin{itemize}[leftmargin=0.5in, topsep=2pt, itemsep=0pt, parsep=0pt] 
\item[$\circ$]
Let $\,\cX^{(2)}_{-1}\, $ be the subgroup 
$\, L^2\bbbz/\frac{L_\tp}{L^2}\bbbz\, \times\, L\bbbz^3/\frac{L_\sp}{L}\bbbz^3\, $ of 
$\,\cX^{(1)}_0\, $ and 

\item[$\circ$]
let  $\, \cQ^{(1)}\, :\, \bbbc^{\cX^{(1)}_0}
\, \rightarrow\, \bbbc^{\cX^{(2)}_{-1}}\, $ 
be a linear `block averaging' operator 
that commutes with complex conjugation. 

\item[$\circ$]
Introduce the unit discrete torus
$\,\cX_0^{(2)}\, 
=\, \bbbz/\frac{L_\tp}{L^4}\bbbz\, \times\, \bbbz^3/\frac{L_\sp}{L^2}\bbbz^3\, $ and 

\item[$\circ$]
the  isomorphism $\, x=(x_0,\bx)\, \in\, \cX_0^{(2)}\, 
\rightarrow\, 
(L^2x_0,\, L\bx)\, \in\, \cX^{(2)}_{-1}\, $.
\end{itemize}
As before, integrate against the normalized Gaussian to obtain the  block spin
transform of $\,F_1\, $ associated to 
$\,\cQ^{(1)}\, $
\begin{equation*}
B_2(\theta_*,\, \theta)
= \sfrac{1}{N^{(1)}}
\int_{\lower3pt\hbox{${\sst\bbbc^{\cX_0^{(1)}}}$} }  
\prod_{x\, \in\, \cX_0^{(1)}}\!\!
            \sfrac{d\psi(x)^*\wedge d\psi(x)}{2\pi\imath}\  
e^{-\frac{1}{L^2}\< \theta_*-\cQ^{(1)}\psi^*\, ,\,
                    \theta -\cQ^{(1)}\psi\>_{-1}
  }
\ F_1(\psi^*,\, \psi)
\end{equation*}
and then rescale to obtain the block spin renormalization group 
transform  
\begin{equation*}
F_2(\psi_*,\, \psi)
=
B_2\big(\bbbs^{-1}\psi_*,\, \bbbs^{-1}\psi\big) 
\end{equation*}
where $\, \big(\bbbs^{-1}\psi\big)(y_0,\by)\, =\, 
L^{-\nicefrac{3}{2}}\psi(\frac{y_0}{L^2},\frac{\by}{L})
\, $ for any $\, \psi\, \in\, \bbbc^{\cX_0^{(2)}}\, $.
Interchanging the order of integration,
\begin{equation*}
L^{-3|\cX_0^{(2)}|}\,  L^{-3|\cX_0^{(1)}|}
\int_{\lower3pt\hbox{${\sst\bbbc^{\cX_0^{(2)}}}$}}  
 \prod_{x\, \in\, \cX_0^{(2)}}\!\!
            \sfrac{d\psi(x)^*\wedge d\psi(x)}{2\pi\imath}\  
\
F_2(\psi^*,\, \psi)
=
\int_{\bbbc^{\cX_0}}
   \prod_{x\, \in\, \cX_0}\!\!
            \sfrac{d\psi(x)^*\wedge d\psi(x)}{2\pi\imath}\  
F_0(\psi^*,\, \psi)
\end{equation*}
We keep repeating the construction  to generate a sequence 
$\,F_n(\psi_*,\, \psi)\, ,\, n\, \ge\, 1\, $,
of functions defined on spaces 
$\,\bbbc^{\cX^{(n)}_0}\!\times\, \bbbc^{\cX^{(n)}_0}\, $.
\cite{PAR1,PAR2} concerns a sequence 
$\,F_n^\SF(\psi_*,\, \psi)\, $ of
`small field' approximations to the $F_n$'s. We expect, and provide 
some supporting motivation for, but do not
prove, that $F_n=F_n^\SF+ O\big(e^{-{1/\fv_0^\veps}}\big)$.
For the precise definition, see  \cite[\S \sectINTstatPhase\ and, 
in particular, Definition \defHTapproximateblockspintr]{PAR1}.
For the supporting motivation see \cite{PARL}.

To make a specific choice for the, to this point arbitrary, 
sequence $\,\cQ^{(0)}$, $\cdots$, $\cQ^{(n)}$, $\cdots$ 
of block averaging operators,
let $\, q(x)\, $ be a nonnegative, compactly supported, even function  on 
$\, \bbbz\times\bbbz^3\, $ and  
$\,Q\, $  the associated convolution operator\footnote{
By abuse of notation, we use the same symbol $Q$ for the convolution operator
acting on all of the spaces $\bbbc^{\cX_0^{(n)}}$.} 
\begin{equation*}
(Q\psi)(y)
=
\sum_{x\in \bbbz\times  \bbbz^3} q(x)\, \psi\big(y+[x]\big)
\ \ ,\ \
\psi\, \in\, \bbbc^{\cX^{(n )}_{0}}
\ \ ,\ \
y\, \in\, \cX^{(n+1)}_{-1}\, \subset\, \cX^{(n)}_0
\end{equation*}
where $\, [x]\, $ is the point in the quotient $\, \cX^{(n)}_0
\, =\, \bbbz/\frac{L_\tp}{L^{2n}}\bbbz\, \times\, \bbbz^3/\frac{L_\sp}{L^n}\bbbz^3\, $
represented by $\,x\in\bbbz\times \bbbz^3\, $. By construction, $\, Q\psi\, \in\, 
\bbbc^{\cX^{(n+1 )}_{-1}}\, $. We fix   
$\,q(x)\, $ to be the convolution of the indicator function of 
the (discrete) rectangle $\,[-\frac{L^2-1}{2}\, ,\, \frac{L^2-1}{2}]
\, \times\, {[-\frac{L-1}{2}\, ,\, \frac{L-1}{2}]}^3\, $ in 
$\, \bbbz\times\bbbz^3\, $ convolved with itself
four times and normalized so that its sum over $\,\bbbz\times\bbbz^3\, $ 
is one. In \cite{PAR1,PAR2} the basic objects are
the `small field' block spin renormalization iterates 
$\,F_n^\SF(\psi_*,\, \psi)\,$, where at each step  $\, Q\, $ 
is chosen to be convolution  with the fixed kernel $\,q\,$.

If we had defined $\, Q\, $ by convolving just with the 
indicator function of the rectangle itself, properly normalized, then
$\,(Q\psi)(y)\, $ would be the usual average 
of $\,\psi(x)\,$ over the rectangular box in $\,\cX^{(n)}_0\, $ 
centered at $\,y\,$ with sides $\,L^2\, $ and $\, L \, $. 
We work with the smoothed averaging kernel rather than the sharp one
for technical reasons: commutators $\,[\partial_\nu,\, Q]\, $
are routinely generated and are small enough when $\,Q\, $ 
is smooth enough. For the rest of this overview we will \emph{pretend} 
that $\, q\, $ is just the indicator function of the rectangle 
and formulate our results as if this were the case.
We will also \emph{pretend} that the operator 
$\,h\, $  on $\,X\, $ appearing in the action $A_0(\psi_*,\psi)$
is (minus) the lattice Laplacian. Full, technically complete, statements are in
\cite[\S\sectINTmainResults]{PAR1}.

\medskip

Our main result is: If $\,\eps\, >\, 0\, $ 
and $\,\fv_0\, $ are small enough and $\,L\, $ is
large enough, there exists a\footnote{An explicit formula for $\mu_*$
is given in \cite[(\eqnINTmustardef)]{PAR1}.} $\,\mu_*\, =\, O(\fv_0)\, $, 
such that for all\footnote{We are weakening some of the statements,
for pedagogical reasons. In particular, the sets of allowed $\mu_0$'s and
$n$'s are a bit larger than the sets specified here.} 
$\,\mu_*\, +\, \fv_0^{\nicefrac{5}{4}} \, <\, \mu_0\, <\, 
                                                 \fv_0^{\nicefrac{9}{10}}\, $ 
and all 
$\, n\, <\, \sfrac{2}{5}\sfrac{\log\nicefrac{1}{\fv_0}}{\log L}\, $,
the `small field approximations' $F_n^\SF$ to the $F_n$'s are
\begin{align*}
F_n^\SF(\psi_*,\, \psi)
&=
\sfrac{1}{\cZ_n}\, 
\exp\big\{-A_n\big(\psi_*,\psi,\, 
       \phi_{* n}(\psi_*,\, \psi),\phi_n(\psi_*,\, \psi)\big)
\, +\, 
p_n(\psi_*,\psi,\nabla\psi_*,\nabla\psi)\big\}
\\
A_n
&=
a_n\<(\psi_*-Q_n\phi_{*n})\,,\, (\psi-Q_n\phi_n)\>_0
\, +\, 
\< \phi_{* n} ,\, (-\partial_0-\Delta) \phi_n\>_n \cr&\hskip2in
\,-\,\mu_n\< \phi_{* n} ,\, \phi_n\>_n
\, +\, 
\cV_n(\phi_{* n} ,\,  \phi_n)
\end{align*}
on the domain
\begin{align*}
S_n
= 
\Big\{(\psi_*,\psi)\, \in\,  \bbbc^{2\cX^{(n)}_0}
\  \Big|\  
&|\psi_{(*)}(x)| \le  \ka_n
\ ,\ 
|\partial_\nu\psi_{(*)}(x)|\, \le  \ka'_n
\ ,
\\&\hskip2in
0\le\nu\le 3\ ,\ x\, \in\, \cX_0^{(n)}
\, \big\}
\end{align*}
and zero on its complement. 
Here,
\begin{itemize}[leftmargin=*, topsep=2pt, itemsep=0pt, parsep=0pt]
\item
you can think of the radii $\ka_n$ and $\ka'_n$ as being roughly
$L^{\sfrac{3}{4}n}\fv_0^{-\sfrac{1}{3}+\eps}$  and
$L^{\sfrac{3}{8}n}\fv_0^{-\sfrac{1}{3}+\eps}$, respectively.
Explicit expressions for 
$\ka_n$ and $\ka'_n$ are given in \cite[Definition \defHTbasicnorm.a]{PAR1}.

\item
$\,\phi_{* n}(\psi_*,\, \psi)\, $ and $\,\phi_n(\psi_*,\, \psi)\, $
are (nonlinear) maps  from an open neighborhood of the origin in 
$\,\bbbc^{\cX_0^{(n)}}\!\!\times\bbbc^{\cX_0^{(n)}}\, $
to $\,\bbbc^{\cX_n}\, $, where 
$\,\cX_n\,$ is the discrete torus, isomorphic to  $\,\cX_0\, $, 
but scaled down to have lattice spacing $\,L^{-2n}\, $ in the time direction and  
$\,L^{-n}\, $ in the spatial directions\footnote{ 
   $\,\cX_n\,=\,  
        \frac{1}{L^{2n}}{\bbbz}/\frac{L_\tp}{L^{2n}}{\bbbz}
        \times
        \frac{1}{L^n}{\bbbz}^3/\frac{L_\sp}{L^n}\, {\bbbz}^3
    \,$
and the map $\, u\in \cX_n\, \mapsto\,  x= (L^{2n}u_0,\, L^n{\bf u})\in \cX_0\, $
is an isomorphism of Abelian groups.
}.
We say more about them in the last of this sequence of bullets. 
Given `external fields' $\,\psi_*,\, \psi\, $, the functions 
$\, \phi_{* n}(\psi_*,\, \psi)(u)\, $,
$\, \phi_n(\psi_*,\, \psi)(u) \, $ on 
$\,\cX_n\, $ are referred to as the ``background fields''
at scale $\, n\, $.

\item $\, \<f\, ,\, g\>_0\, =\!\!  
           \sum\limits_{x\, \in\,\cX_0^{(n)} }\!\!f(x)g(x)\, $
and
$\, \<f\, ,\, g\>_n\, =\,  
L^{-5n}\!\!\sum\limits_{u\, \in\,\cX_n }\!\!f(u)g(u)\, $
are the natural real inner products on 
$\, \bbbc^{\cX_0^{(n)}}\, $ and
$\, \bbbc^{\cX_n}\, $.

\item
$\, Q_n\, :\, \bbbc^{\cX_n}\, \rightarrow\, \bbbc^{\cX_0^{(n)}}\, $ 
is the linear map for which $\,(Q_nf)(x)\, $ is the average of
$\,f\, \in\, \bbbc^{\cX_n }\, $ over the square box in $\, \cX_n\, $ 
centered at $\, x\in \cX_0^{(n)}\, $ with sides $1$. (This
box contains $L^{2n}\times {(L^n)}^3$ points of $\cX_n$.)

\item
$a_n\, =\, \frac{1-L^{-2}}{1-L^{-2n}} $

\item 
$\, -\partial_0-\Delta \, $
is the natural heat operator on the `fine' discrete torus 
$\,\cX_n\, $.

\item
For each $f_*,\, f\, \in\, \bbbc^{\cX_n}$
\begin{equation*}
\cV_n(f_*,f) 
\, =\,
\sfrac{1}{2}\,  
\big(\sfrac{1}{L^{5n}}\big)^4 \!\!
\sum\limits_{ \atop {u_j\, \in\,  \cX_n }
                     {j=1,2,3,4}}
\,
V_n(u_1,u_2,u_3,u_4)
\,
f_*(u_1) f(u_2)
f_*(u_3) f(u_4)
\end{equation*}
where $V_n(u_1,u_2,u_3,u_4)$ is close to
\begin{equation}\label{eqnOVVnu}
V_n^{(u)}(u_1,u_2,u_3,u_4)
=\sfrac{1}{L^n}{(L^{5n})}^3\  V_0(U_1,U_2,U_3,U_4)
\ ,\quad U_j\, =\, (L^{2n}\, u_{j0},\, L^n \bu_j)
\end{equation}

\item
The perturbative correction 
$\,p_n\big(\psi_*,\psi,\, \{\psi_{*\nu}\}_{\nu=0}^3,\, 
                \{\psi_\nu\}_{\nu=0}^3 \big)\, $ 
is a power series in the ten variables
$\,\psi_*,\, \psi,\, \{\psi_{*\nu},\, \psi_\nu\}_{\nu=0}^3
\, \in\, \bbbc^{\cX^{(n)}_0}\, $, 
with no $\,\psi_*(x)\psi(y)\, $ or constant terms, such that each 
nonzero term has as many factors with asterisks as factors without asterisks.  
It converges\footnote{
It is necessary to measure the size of $\,p_n\, $ by introducing
an appropriate norm. See the last paragraphs of this overview.} when 
$\,|\psi_{(*)}(x)|\, \le \ka_n\, $ 
and
$\,| \psi_{(*)\nu}(x)|\, \le 
\ka'_n$
for all $0\le\nu\le 3$ and $\,x\, \in\, \cX_0^{(n)}\,$.

\item $\cZ_n$ is a normalization constant\footnote{When we 
take logarithms and ultimately differentiate with respect to 
an external field to obtain correlation functions, it will disappear.}.

\item $\,\mu_n\, $ is the `renormalized' chemical potential\footnote{We 
will describe the  inductive construction of $\,\mu_n\, $ later on in 
this overview. The dependence of $\,p_n\, $ on the derivatives of the fields 
arises because of the renormalization of the chemical potential.
}. It is close to $\,L^{2n}\mu_0\, $.

\item 
For each pair in the polydisc  
\begin{equation*}
\Big\{\, (\psi_*,\, \psi)\, \in\, 
\bbbc^{\cX_0^{(n)}}\times \bbbc^{\cX_0^{(n)}}
\, \Big|\, |\psi_{(*)}(x)|\, \le \ka_n \ 
\text{for all}\   x\, \in\, \cX_0^{(n)}\, \Big\}
\end{equation*}
the fields 
$\,\phi_{*n}(\psi_*,\, \psi)(u)\, $,
$\,\phi_n(\psi_*,\, \psi)(u)\, $ on 
$\, \cX_n\, $
are  critical points of the functional
\begin{align*}
(\phi_*,\, \phi) 
\  
\mapsto
\ 
&A_n(\psi_*,\psi,\phi_*,\phi)\\
&=a_n\<(\psi_*-Q_n\phi_*)\,,\, (\psi-Q_n\phi)\>_0
- \< \phi_* ,\, (\partial_0+\Delta+\mu_n) \phi)\>_n
 + \cV_n(\phi_* ,  \phi)
\end{align*}
The maps $\, \phi_{* n}(\psi_*,\, \psi)\, $,
$\, \phi_n(\psi_*,\, \psi)\, $
are holomorphic on that polydisc. 
\end{itemize}

\bigskip
In practical terms, what have we achieved? If $\,\psi\, =\, z\, $
is a constant field on $\, \cX_0\, $, then the dominant part of the 
initial effective potential is
\begin{align*}
A_0(\psi^*,\, \psi)
&= \cV_0(\psi^*,\, \psi)\, -\, \mu_0\, \<\psi^*,\, \psi\>_0
= |\cX_0|\ \Big(\frac{\rv_0}{2}|z|^4-\, \mu_0|z|^2\Big) \cr
&= |\cX_0|\, \frac{\rv_0}{2}\, 
   \Big[ \Big(|z|^2-\frac{\mu_0}{\rv_0}\Big)^2
   -\, \frac{\mu_0^2}{\rv_0^2}\,\Big]
\end{align*}
where, $\, \rv_0\, =\,  \sum\limits_{ \cX_0^3 }
\, V_0(0,x_2,x_3,x_4)\, $. The graph of the real valued function
\begin{equation*}
\frac{\rv_0}{2}\, 
\Big[ \Big(|z|^2-\frac{\mu_0}{\rv_0}\Big)^2
-\, \frac{\mu_0^2}{\rv_0^2}\,\Big]
\end{equation*}
over the complex plane $\, z\, =\, x_1+\imath x_2\, $
is a surface of revolution around the $\, x_3$--axis 
with the circular well of absolute minima $\,|z|\, =\,
\sqrt{\frac{\mu_0}{\rv_0}}\, $. Our hypothesis on $\, \mu_0\, $
implies that the radius and depth of the well are  of order one
and order $\, \rv_0\, $ respectively.
After $\,n\, $ renormalization group steps,  the effective potential becomes
\begin{equation}\label{eqnOVeffectPot}
A_n\big(\psi^*,\psi,\, 
\phi_{* n}(\psi^*,\, \psi),\phi_n(\psi^*,\, \psi)\big)
\, \big|_{\psi = z}
\approx\ 
|\cX_0^{(n)}|\, \frac{\rv_0}{2L^n} 
\Big[\Big(|z|^2-\frac{\mu_n}{\nicefrac{{\rv_0}}{L^n}}\Big)^2
-\, \frac{\mu_n^2}{(\nicefrac{{\rv_0}}{L^n})^2}\Big]
\end{equation}
since, by \cite[Remark \remBGEcnstFld]{BGE}, 
$\, \phi_n(\psi^*,\, \psi)\, |_{\psi\, =\, z}\, \approx\, z\, $
and 
$\,\phi_{* n}(\psi^*,\, \psi)\, |_{\psi\, =\, z}\, \approx\, z^*\, $. 
The graph is again a surface of revolution 
with the circular well of absolute minima 
$\, |z|\, =\,
\sqrt{\frac{\mu_n}{\nicefrac{\rv_0}{L^n}}}\  $,
but now the radius and depth are  of order $\,L^{\frac{3}{2}n}\, $ 
and order $\, L^{5n}\rv_0\, $ respectively; the well
is developing. We stop the flow  when the well
becomes so wide and so deep that we can no longer construct background
fields by expanding around $\, \psi_*,\, \psi\, =\, 0\, $.
This happens as $\,\mu_n\, $ approaches order one.

If the power series expansion of the perturbative correction
$\, p_n\, $ had a quadratic part 
$\, \sum\limits_{x,y\in\cX_0^{(n)}}\hskip-5pt K(x,y)\, \psi_*(x)\psi(y)\, $ 
the discussion of the evolving well in the last paragraph 
would be misleading, because the minimum of the total action 
$\,A_n - p_n\, $ would not be close enough
to the minimum of the dominant part $\,A_n\, $.
The requirement that $\,p_n\,$ must not contain quadratic terms is 
the  renormalization condition for the chemical potential.
(See, Step 9 below.) 
Under the scaling map \eqref{eqnOVscalingMap}, the local monomials
\begin{equation*}
\<\psi_*,\psi\>_0\qquad \<\psi_*,\partial_\nu \psi\>_0\ 1\le\nu\le 3
\end{equation*}
are relevant and the local monomials
\begin{equation*}
\<\psi_*,\partial_0 \psi\>_0\qquad 
\<\partial_\nu\psi_*,\partial_{\nu'} \psi\>_0\ 1\le\nu,\nu'\le 3
\end{equation*}
are marginal. The local monomials $\<\psi_*,\partial_\nu \psi\>_0$,
 $1\le\nu\le 3$, do not appear, because of reflection invariance.
See \cite[Definition \defSYfullSymmetry\ and Lemma \lemSYfullsymmetry]{PAR1}.
So $p_n$ does not contain any relevant monomials.

The parabolic renormalization group flow drives the system away from 
the trivial (noninteracting) fixed point. To continue,  we will 
have to construct background fields by expanding about configurations 
supported 
near the bottom of the developing well, analogously to the 
`Bogoliubov Ansatz'. At present, we expect to continue the parabolic
flow, but expanding about configurations supported 
near the bottom of the well, through a transition regime  
(which overlaps with the regime of \cite{PAR1,PAR2}) until $\mu_n$ becomes
large enough (but still of order one), and then switch to a 
new `elliptic' renormalization
group flow for the push to the symmetry broken, superfluid fixed point.
In Appendix \ref{sectTransition}, below, we perform several model 
computations that contrast the parabolic nature of the early 
renormalization group steps with the elliptical nature of the 
late renormalization group steps. 

\medskip
The next part of this overview is an outline, in nine steps, of the 
inductive construction that uses a steepest descent/stationary phase 
calculation to build the desired form for 
$\,F_{n+1}(\psi_*,\, \psi)\, =\, 
  B_{n+1}\big(\bbbs^{-1}\psi_*,\, \bbbs^{-1}\psi\big) 
\, $, 
from that of $\,F_n(\psi_*,\, \psi)\ ,\ n\, \ge\, 0\, $, 
where
\begin{equation*}
B_{n+1}(\theta_*,\, \theta)
=
\sfrac{1}{N^{(n)}}\int_{\lower3pt\hbox{${\sst\bbbc^{\cX_0^{(n)}}}$}}  
\prod_{x\, \in\, \cX_0^{(n)}\!\!\!\!}\!\!
            \sfrac{d\psi(x)^*\wedge d\psi(x)}{2\pi\imath}\  
e^{-\frac{1}{L^2}\< \theta_* \, -\, Q\psi^*\, ,\, \theta\, -\, Q\psi\>_{-1} }
\
F_n(\psi^*,\, \psi)
\end{equation*}
We are expecting that, by induction,
\begin{equation}\label{eqnOVbnplusone}
\begin{split}
&B_{n+1}(\theta_*,\, \theta)
= \sfrac{1}{N^{(n)}}
\int_{S_n}\hskip-3pt \prod_{x\, \in\, \cX_0^{(n)}\hskip-7pt}\hskip-3pt
            \sfrac{d\psi(x)^*\wedge d\psi(x)}{2\pi\imath}\ 
e^{-\frac{1}{L^2}\< \theta_* - Q\psi^*\, ,\, \theta-Q\psi\>_{-1}}
F_n^\SF(\psi^*,\, \psi)
+O\big(e^{-{1/\fv_0^\veps}}\big)
\\
&\hskip0.1in= \sfrac{1}{N^{(n)}\cZ_n}
\int_{S_n}\hskip-2pt \prod_{x\, \in\, \cX_0^{(n)}\!\!\!\!}\!\!
            \sfrac{d\psi(x)^*\wedge d\psi(x)}{2\pi\imath}\  
\, e^{
-\frac{1}{L^2}\< \theta_* -Q\psi^*\, ,\, \theta- Q\psi\>_{-1}
\, -\, A_n(\psi^*,\psi,\, \phi_{* n},\phi_n)
\, +\, 
p_n}
\  +\   
O\big(e^{-{1/\fv_0^\veps}}\big)
\\
\noalign{\vskip4pt}
&\hskip0.1in=\hskip8pt  
\text{Dominant Part}\hskip8pt +\hskip8pt \text{Non Perturbative Correction}
\end{split}
\end{equation}
We emphasise that Steps 1 and 6, which control the difference between 
$F_{n+1}(\psi_*,\psi)$ and its, dominant, `small field', part
$F_{n+1}^\SF(\psi_*,\psi)$, have not been proven, though we do supply
some motivation in \cite{PARL}.

\Item\textbf{Step 1} (Large field generates small factors). 
If  $\Psi\!\in\!\bbbc^{\cX_0^{(n+1)}}$ is `large field', that is
$\Psi\!\notin\! S_{n+1}$, then we expect 
$\, B_{n+1}(\bbbs^{-1}\Psi^*,\bbbs^{-1}\Psi)
=
O\big(e^{-{1/\fv_0^\veps}}\big)
\, $,
since the real part of the exponent appearing in the integrand 
of \eqref{eqnOVbnplusone} is of order 
$\,-\sfrac{1}{\fv_0^\veps} 
\, $. See \cite[Proposition \propALFfirst, ``Corollary'' \corALFfirst\
and the subsequent \emph{Steps} 1 and 2]{PARL}.
\medskip

We   assume for the rest of the discussion that 
$\, \Psi\, \in\, \bbbc^{\cX_0^{(n+1)}}\, $
is `small field', that is
$\,\Psi\, \in\, S_{n+1}\, $, and therefore construct holomorphic 
functions of $\,(\Psi_*,\Psi)\, $ on the product  
$\,S_{n+1}\times S_{n+1}\, $. Let $\th_*=\bbbs^{-1}\Psi_*$ and 
$\th=\bbbs^{-1}\Psi$.

\Item\textbf{ Step 2} (Holomorphic form representation). 
We wish to analyze the integral in \eqref{eqnOVbnplusone}
by a steepest descent/stationary phase argument. Recall that
a critical point of a function $f(z)$ of one complex variable
$z=x+iy$, that is \emph{not} analytic in $z$,  is a point
where both partial derivatives $\sfrac{\partial f}{\partial x}$
and $\sfrac{\partial f}{\partial y}$, or equivalently,
both partial derivatives
$\sfrac{\partial f}{\partial z}=\sfrac{1}{2}
   \big(\sfrac{\partial\hfill}{\partial x} 
   -i\sfrac{\partial\hfill}{\partial y}\big)f$ and
$\sfrac{\partial f}{\partial \bar z}=\sfrac{1}{2}
   \big(\sfrac{\partial\hfill}{\partial x}
    +i\sfrac{\partial\hfill}{\partial y}\big)f$
vanish. We prefer the latter formulation. So we rewrite
the integral in \eqref{eqnOVbnplusone} in a form that allows
us to treat $\psi$ and its complex conjugate as independent
fields.
For each fixed
$\, (\theta_*,\theta)\, $, the `action'
\begin{equation}\label{eqnOVeffAct}
\begin{split}
 \act_n(\theta_*,\theta,\,\psi_*,\psi)
&=
\Big\{\!
-\sfrac{1}{L^2}\< \theta_*- Q\psi_*\, ,\, \theta-Q\psi\>_{-1}
\, -\, A_n\big(\psi_*,\psi,\phi_{* n}(\psi_*,\psi),\phi_n(\psi_*,\psi)
\big)\!
\Big\}
\\&\hskip3.0in
\, +\, p_n(\psi_*,\psi,\nabla\psi_*,\nabla\psi) \\
&= 
-A_{n,\eff}(\theta_*,\theta,\,\psi_*,\psi)\  
  +\ p_n(\psi_*,\psi,\nabla\psi_*,\nabla\psi) 
\end{split}
\end{equation}
is a holomorphic function of $\, (\psi_*,\psi)\, $
on $\,S_n\times S_n\, $. By design,
the Dominant Part of $\,B_{n+1}(\theta_*,\, \theta)\, $
in \eqref{eqnOVbnplusone} is expressed as (a constant times)
the integral of the holomorphic form 
\begin{equation}\label{eqnOVform}
e^{\act_n(\theta_*,\theta,\, \psi_*,\psi)} 
\bigwedge_{x\, \in\, \cX^{(n)}_0}
\frac{d\psi_*(x)\wedge  d\psi(x)}{2\pi\imath}
\end{equation}
of degree $\,2|\cX^{(n)}_0|\, $
over the real subspace in $\,S_n\times S_n\, $ 
given by $\,\psi_*\, =\, \psi^*\, $. We shall see below that,
typically, the critical point does not lie in the real subspace
and so is not in the domain of integration. This representation permits
us to use Stokes' theorem\footnote{The argument is similar to the use
of Cauchy's theorem in stationary phase arguments for functions  of
one variable.}, 
to shift the contour of integration to a non real contour that 
does contain the critical point of (the principal 
terms of) the action. The shift will be implemented in Step 6.

\Item\textbf{Step 3} (Critical Points). 
Our next task is to find critical points. In \eqref{eqnOVeffAct}, above,
we wrote the exponent, $\act_n(\theta_*,\theta,\,\psi_*,\psi)$,
as the sum of a very explicit, main, part 
$-\,A_{n,\eff}\, $ and a  not very explicit, smaller, part $p_n$.
We just find the critical points of $A_{n,\eff}$ rather than
the full $\act_n$.
Indeed, there is a unique pair of holomorphic maps\footnote
{In \cite{PAR1,PAR2} these maps are called $\,\psi_{n*} \, ,\, \psi_n \, $.} 
$\, 
\psi_{*\mathrm{cr}}(\theta_*,\, \theta)
\,$, 
$\,\psi_\mathrm{cr}(\theta_*,\, \theta)
\, $ 
from $\,(\bbbs^{-1}S_{n+1})\, \times\,  (\bbbs^{-1}S_{n+1})\, $
to $\,S_n\,$
such that the gradient 
$\, \big(\begin{matrix}
                  {\sst\nabla_{\!\psi_*}} \\
                  \noalign{\vskip-4pt}
                  {\sst\nabla_{\!\psi_{\phantom{*}}}}
         \end{matrix}
\big)\, $ 
of $\,A_{n,\eff}(\theta_*,\theta,\,\psi_*,\psi)\, $  vanishes when 
$\, \psi_*\, =\, \psi_{*\mathrm{cr}}(\theta_*,\, \theta)\, ,\,  
\psi \, =\, \psi_{\mathrm{cr} }(\theta_*,\, \theta)\, $.
This pair of  `critical field maps' can be constructed by solving
the critical point equations, a nonlinear parabolic system of (discrete)
partial difference equations, using the natural contraction mapping argument 
to perturb off of the linearized equations\footnote{In \cite{PAR1,PAR2,BGE}
we take another route to the critical field maps. The background fields 
$\phi_{(*)n}(\psi_*,\psi)$ are constructed first, using the natural
contraction mapping argument to perturb off of the linearized background
field equations. See \cite[Proposition \propBGEphivepssoln]{BGE}. 
The critical fields can then be expressed as functions of the background fields.
See \cite[Proposition \propBGAomnibus]{PAR1}.  }. The analysis of the linearized
equations is based on a careful examination of some linear operators given
in \cite{POA}. Beware that, in general, 
$\psi_{*\mathrm{cr}}(\theta^*,\theta)\ne {\psi_{\mathrm{cr}}(\theta^*,\theta)}^*$.

\medskip
To start the stationary phase calculation, we  factor  
the integral of the holomorphic form \eqref{eqnOVform}
over the real subspace $\set{(\psi_*,\psi)\in S_n\!\times\! S_n}{\psi_*=\psi^*}$ 
as the product of 
\begin{equation*} 
e^{\act_n(\theta_*,\theta,\, 
          \psi_{*\mathrm{cr}}(\th_*,\th),\psi_{\mathrm{cr}}(\th_*,\th)\,)}
\end{equation*}
and the `fluctuation integral'
\begin{equation}\label{eqnOVfluctInt}
\int_{\atop{ \text{real subspace of}}{S_n\times S_n} }
\hskip-6pt
e^{\act_n(\theta_*,\theta,\, \psi_*,\psi)\, 
-\, \act_n(\theta_*,\theta,\, 
           \psi_{*\mathrm{cr}}(\th_*,\th),\psi_{\mathrm{cr}}(\th_*,\th)\,)} 
\bigwedge\limits_{x\, \in\, \cX^{(n)}_0}
                  \frac{d\psi_*(x)\wedge d\psi(x)}{2\pi\imath}
\end{equation}
 
\Item \textbf{Step 4} (The Value of the Action at the Critical Point).
We would expect that the biggest contribution to the integral
would come from simply evaluating the exponent at the critical
point, and that the biggest contribution to the value of the exponent
$\act_n$ at the critical point would come from evaluating $-A_{n,\eff}$
at the critical point. By \cite[Proposition \propBGAomnibus.c]{PAR1}
\begin{align*}
&
A_{n,\eff}(\theta_*,\theta,\, \psi_*,\psi)\, 
\big|_{\psi_*=\psi_{*\mathrm{cr}}(\theta_*,\, \theta) \, ,\,  
       \psi  = \psi_{\mathrm{cr}}(\theta_*,\, \theta)} \\
&\hskip2in=
{\check A}_{n+1}\big(\theta_*,\, \theta,\, 
{\check\phi}_{* n+1}(\theta_*,\, \theta),\, {\check\phi}_{n+1}(\theta_*,\, \theta)\big)
\end{align*}
where
\begin{align*}
{\check A}_{n+1}(\theta_*,\, \theta,\, f_*,\, f)
&=
\sfrac{a_{n+1}}{L^2} 
\< \theta_*- Q Q_nf_*\,,\,\theta-QQ_nf\>_{-1}
- \< f_*\, ,\, (\partial_0+\Delta+\mu_n) f \>_n
\\&\hskip3.3in
+   \cV_n(f_* ,\,  f)
\end{align*}
and the `checked' field 
\begin{equation*}
{\check\phi}_{(*) n+1}(\theta_*,\, \theta)
=
\phi_{(*) n }\big(
        \psi_{*\mathrm{cr}}(\th_*,\th,  \psi_\mathrm{cr}(\th_*,\th)
        \big)
\end{equation*}
is the background field evaluated at the critical point.
Consequently, 
\begin{align*}
e^{\act_n(\theta_*,\theta,\, 
     \psi_{*\mathrm{cr}}(\theta_*,\, \theta),\psi_{\mathrm{cr} }(\theta_*,\, \theta))}
=
e^{ -{\check A}_{n+1}(\theta_*,\, \theta,
\, {\check\phi}_{* n+1}(\theta_*,\, \theta),\, 
   {\check\phi}_{n+1}(\theta_*,\, \theta))
\, +\, 
p_n(\psi_{*\mathrm{cr}},\psi_{\mathrm{cr} },
    \nabla\psi_{*\mathrm{cr}},\nabla\psi_{\mathrm{cr} })}
\end{align*}

\noindent\textbf{Remark}
Bear in mind that the checked fields depend implicitly on $\,\mu_n\, $. 
In the next steps, we will build a  new `renormalized' chemical potential 
$\,\mu_{n+1}\, $ that will appear in $\,A_{n+1}\, $.
If, for the purposes of discussion, we ignored the effects of renormalization, 
$\,A_{n+1}\,$
would just be a rescaled $\,{\check A}_{n+1}\, $
(see \cite[Definition \defSCacheck\ and Lemma \lemSCacheckOne.c]{PAR1})
and the new background field $\,\phi_{(*) n+1}\, $ would just
be a rescaled $\,{\check\phi}_{(*) n+1}\, $
(see \cite[Definition \defBGAphicheck\ and Proposition \propBGAomnibus.b]{PAR1}).
So, we are not far off.

\Item\textbf{Step 5} (Diagonalization of the Quadratic Form in the Fluctuation Integral).
Next consider the fluctuation integral \eqref{eqnOVfluctInt}.
Make the  change  variables 
\begin{equation*}
(\psi_*,\psi)
\ \rightarrow\ 
\big(\delta\psi_*\, =\, \psi_*-\psi_{*\mathrm{cr}}(\theta_*,\theta)
\, ,\, \delta\psi\, =\, \psi-\psi_{\mathrm{cr}}(\theta_*,\theta)\big)
\end{equation*}
to shift the critical point to $\de\psi_*=\de\psi=0$. Substitute
$\psi_{(*)}=\psi_{(*)\mathrm{cr}}(\theta_*,\theta) +\de\psi_{(*)}$ into the main
part
\begin{equation*}
A_{n,\eff}(\theta_*,\theta,\, \psi_*,\psi)\, -\, 
A_{n,\eff}\big(\theta_*,\theta,\, 
          \psi_{*\mathrm{cr}}(\theta_*,\theta),\psi_{\mathrm{cr} }(\theta_*,\theta)
        \big)
\end{equation*}
of the exponent and expand in powers of $\de\psi_{(*)}$. The
constant and, by criticality, linear parts vanish. The 
quadratic term has a dominant part (see \cite[Lemma \lemSTdeA, (\eqnSTdeAuncheck)]{PAR1} and \cite[Lemma \lemOSFmainlem]{PAR2}), 
that is independent of $\,\theta_*,\theta\, $.
All of the eigenvalues of the kernel of that dominant part are bounded away 
from the negative real axis, uniformly in $\,n\, $.\footnote{A major part of \cite{POA}
is devoted to proving this vital technical statement.}
So, it is invertible and its inverse, $C^{(n)}$, has a square root
$\,D^{(n)}\, $, all of whose eigenvalues have strictly positive real parts. 
See \cite[Corollary \corPOCsquareroot]{POA}.
Now, the Taylor expansion of the above difference of effective actions 
in the new variables 
$\, 
\delta\psi_* \, =\, D^{(n)^T}\zeta_*
\, ,\, 
\delta\psi\, =\, D^{(n)}\zeta
\, $
becomes 
\begin{equation*}
\<\zeta_*,\zeta\>_0  
\,+\,  \text{smaller terms of degree } 2\text{ in }\zeta_*,\zeta
\,+\,  \text{terms of degree at least } 3\text{ in }\  \zeta_*,\zeta
\end{equation*}
and the fluctuation integral \eqref{eqnOVfluctInt} becomes
\begin{equation}\label{eqnOVfluctIntB}
\int_{\Omega_n(\theta_*,\, \theta)}
e^{-\<\zeta_*,\zeta\>_0    +  
q_n(\theta_*,\, \theta,\, \zeta_*\, ,\, \zeta)} 
\ \
\det(D^{(n)})^2\hskip-8pt
\bigwedge\limits_{x\, \in\, \cX^{(n)}_0}\hskip-8pt
\frac{d\zeta_*(x)\wedge d\zeta(x)}{2\pi\imath}
\end{equation}
where the domain of integration 
$\,\Omega_n(\theta_*,\, \theta)\, $
consists of the set of all pairs 
$\,(\zeta_*,\zeta)\, \in\, \bbbc^{\cX_0^{(n)}}\times \bbbc^{\cX_0^{(n)}}\, $
such that $\,\big(\psi_{*\mathrm{cr}}(\theta_*,\, \theta)+D^{(n)^T}\!\zeta_*
\, ,\,  \psi_{\mathrm{cr} }(\theta_*,\, \theta)+D^{(n)}\zeta\big)\, $
is in the real subspace of $\,S_n\times S_n\, $.
The term $\,q_n\, $ is holomorphic on the complex domain of all quadruples 
$\,(\theta_*,\, \theta,\, \zeta_*\, ,\, \zeta)\, $
with $\,(\theta_*,\, \theta) \in 
(\bbbs^{-1}S_{n+1})\times(\bbbs^{-1}S_{n+1})\, $ 
and
\begin{equation*}
\big(\psi_{*\mathrm{cr}}(\theta_*,\, \theta)+D^{(n)^T}\!\zeta_*
\, ,\,  \psi_{\mathrm{cr} }(\theta_*,\, \theta)+D^{(n)}\zeta\big)
\in S_n\times S_n
\end{equation*}
See \cite[(\eqnOSAfluctInt) and Corollary \corSTmainCor]{PAR1}.

\Item\textbf{Step 6} (Stokes' Theorem). 
For each pair $\, (\theta_*,\, \theta)\, \in\, 
(\bbbs^{-1}S_{n+1})\, \times\,  (\bbbs^{-1}S_{n+1})
\, $, we construct, in \cite[following (\eqnSFCbilinearform)]{PARL}, 
a  $\,2|\cX_0^{(n)}|+1\, $ (real) dimensional ``cylinder'',
inside the $\,(\zeta_*,\zeta)\, $  domain of analyticity of 
$q_n$, whose boundary consists of 
\begin{itemize}
\item[$\circ$]
the original domain of integration 
$\,\Omega_n(\theta_*,\, \theta)\, $
(which typically does not contain the critical point $\ze=\ze_*=0$),
\item[$\circ$] the desired new domain of integration
\begin{equation*}
D_n
=
\big\{\, (\zeta_*,\zeta)\, \big|\, \zeta_*\, =\, \zeta^*
                                \, ,\, |\zeta(x)|\, <\, 
             \sfrac{1}{4}\big(\sfrac{L^{n+1}}{\fv_0}\big)^{\nicefrac{\veps}{2}}
\, \text{for all}\ x\in\cX_0^{(n)}
\, \big\}
\end{equation*}
(which does contain the critical point $\ze=\ze_*=0$)
\item[$\circ$]
and components on which 
$\, 
e^{-\, \<\zeta_*,\zeta\>_0  \,   +\,  
q_n(\theta_*,\, \theta,\, \zeta_*\, ,\, \zeta)} 
\, $
is 
$\,O(e^{-1/\fv_0^\veps})\, $. 
\end{itemize}
See \cite[(\eqnSFCbilinearformt)]{PARL}.
The  holomorphic differential form in Step 5  
has maximal rank and is therefore closed.
It follows from Stokes' theorem that the 
fluctuation integral \eqref{eqnOVfluctIntB} is equal to the 
small field contribution
\begin{equation}\label{eqnOVfluctIntC}
\FlInt_n(\theta_*,\, \theta)
= \det(D^{(n)})^2
\, 
\int\limits_{D_n} \prod_{x\, \in\, \cX_0^{(n)}\hskip-7pt}\hskip-3pt
            \sfrac{d\ze(x)^*\wedge d\ze(x)}{2\pi\imath}\ 
e^{-\, \<\zeta^*,\zeta\>_0  \,   +\,  
              q_n(\theta_*,\, \theta,\, \zeta^*\, ,\, \zeta)} 
\end{equation}
plus corrections that are expected to be nonperturbatively small.

\Item{\bf Step 7} (The Logarithm of the Fluctuation Integral). 
In \cite{CPC} we developed a simple variant of the polymer expansion
that can be directly applied to the integral in \eqref{eqnOVfluctIntC}
to obtain the logarithm 
$
{\rm Log}\Big[\frac{\FlInt_n(\th_*,\, \th)}
                    {\FlInt_n(0,\, 0)}\Big]
$ 
as an analytic function on
$
(\bbbs^{-1}S_{n+1})\times(\bbbs^{-1}S_{n+1})
$. See \cite[Proposition \propOSFmainprop]{PAR2}.

\Item{\bf Step 8} (Rescaling). 
To this point we have determined that the small field part of 
$B_{n+1}(\th_*,\th)$ is a constant times the exponential of the sum of
\begin{itemize}[leftmargin=*, topsep=2pt, itemsep=0pt, parsep=0pt]
\item[$\circ$] 
the contribution which comes from simply evaluating
$\cA_n$ at the critical point --- in Step 4 we saw that this was
\begin{equation*}
-{\check A}_{n+1}(\theta_*,\, \theta,
\, {\check\phi}_{* n+1}(\theta_*,\, \theta),\, 
   {\check\phi}_{n+1}(\theta_*,\, \theta))
\, +\, 
p_n(\psi_{*\mathrm{cr}},\psi_{\mathrm{cr} },
    \nabla\psi_{*\mathrm{cr}},\nabla\psi_{\mathrm{cr} })
\end{equation*}
\item[$\circ$]
and an analytic function that came, in Step 7, from the fluctuation integral. 
\end{itemize}
We are now ready to scale to get the small field part of
\begin{equation*}
F_{n+1}(\Psi_*,\, \Psi)\, =\, 
  B_{n+1}\big(\bbbs^{-1}\Psi_*,\, \bbbs^{-1}\Psi\big) 
\end{equation*}
Using that
\begin{equation}\label{eqnOVscalingRules}
\begin{split}
\sfrac{1}{L^2}\<\bbbs^{-1}\Psi_*\,,\,\bbbs^{-1}\Psi\>_{-1} &= \<\Psi_*\,,\,\Psi\>_0
\\
\bbbs QQ_n\bbbs^{-1}&=Q_{n+1} 
\\
\<\bbbs^{-1}f_*\,,\,\bbbs^{-1}f\>_n &= L^2\<f_*\,,\,f\>_{n+1}\quad 
\\
\<\bbbs^{-1}f_*,(\partial_0\!+\!\De)\bbbs^{-1}f\>_n &= \<f_,,(\partial_0\!+\!\De)f\>_{n+1}
\end{split}
\end{equation}
(see \cite[Remark \remSCscaling.c and Lemma \lemSCacheckOne.a,b]{PAR1})
we have that
\begin{align*}
&{\check A}_{n+1}(\theta_*,\, \theta,
\, {\check\phi}_{* n+1}(\theta_*,\, \theta),\, 
   {\check\phi}_{n+1}(\theta_*,\, \theta))
            \Big|_{\theta_{(*)}=\bbbs^{-1}\Psi_{(*)}} \\
&\hskip2in=A'_{n+1}(\Psi_*, \Psi\,,
                \,\phi'_{* n+1}(\Psi_*,\Psi)\,,\,  \phi'_{n+1}(\Psi_*,\Psi))
\end{align*}
where
\begin{align*}
A'_{n+1}\big(\Psi_*, \Psi,f_*, f\big)
&=
{\check A}_{n+1}\big(\bbbs^{-1}\Psi_*,\, \bbbs^{-1}\Psi,
\,\bbbs^{-1}f_*,\, \bbbs^{-1}f\big)
\\
&=
a_{n+1} 
\<\Psi_*-  Q_{n+1}f_*\,,\, \Psi- Q_{n+1}f\>_0 
\,  -\, 
\< f_*\, ,\, (\partial_0\!+\!\Delta\!+\!L^2\mu_n) f \>_{n+1} \\ &\hskip3.2in
\,+\, \cV'_{n+1}(f_* ,\,  f)
\\
\phi'_{(*) n+1}(\Psi_*,\Psi)
&=
\bbbs\, {\check \phi}_{(*) n+1}(\bbbs^{-1}\Psi_*,\, \bbbs^{-1}\Psi) 
\end{align*}
and if the kernel $V_n$ of $\cV_n$ were exactly the $V_n^{(u)}$ 
of \eqref{eqnOVVnu}, then the kernel of $\cV'_{n+1}$ would be 
exactly $V_{n+1}^{(u)}$. See \cite[Remark \remSCscaling.h]{PAR1}.
Renormalization is going to tweak, for example, the value of the chemical
potential. As a result $A'_{n+1}$ is not quite $A_{n+1}$ and 
$\phi'_{(*) n+1}$ is not quite $\phi_{(*) n+1}$.
That's the reason for putting the primes on.

  Similarly, the contributions from 
$p_n(\psi_{*\mathrm{cr}},\psi_{\mathrm{cr} },
    \nabla\psi_{*\mathrm{cr}},\nabla\psi_{\mathrm{cr} })$ and from the
fluctuation integral get scaled to
\begin{align*}
&p_{n+1}'\big(\Psi_*,\Psi_ , \{\Psi_{*\nu}\}_{\nu=0}^3 , \{\Psi_\nu\}_{\nu=0}^3 \big)
=
\bigg[p_n\big(\psi_{*\mathrm{cr}}(\th_*,\th) ,\psi_{\mathrm{cr} }(\th_*,\th) ,
          \nabla\psi_{*\mathrm{cr}}(\th_*,\th),
          \nabla\psi_{\mathrm{cr}(\th_*,\th) }\big)
\\&\hskip3in
\ +\ 
{\rm Log}\Big(\frac{\FlInt_n(\bbbs^{-1}\Psi_*,\, \bbbs^{-1}\Psi)}
                    {\FlInt_n(0,\, 0)}\Big)
\bigg]_{\th_{(*)}=\bbbs^{-1}\Psi_{(*)}}
\end{align*}
and we have that, renaming $\Psi_{(*)}$ to $\psi_{(*)}$,
the small field part of $F_{n+1}(\psi_*,\psi)$ is
\begin{align*}
F_{n+1}^\SF(\psi_*,\psi) 
=
e^{-A'_{n+1}(\psi_*, \psi\,,
                \,\phi'_{* n+1}(\psi_*,\psi)\,,\,  \phi'_{n+1}(\psi_*,\psi))
   \, +\, p_{n+1}'(\psi_*,\psi_ ,\nabla\psi^* ,\nabla\psi)
  }
\end{align*}
on $\,S_{n+1}\times S_{n+1}\,$.


\Item\textbf{Step 9} (Renormalization of the Chemical Potential).
At this point, we are close to the end of the induction step, but not
there yet because the power series $\,p'_{n+1}\, $ contains 
(renormalization group) relevant contributions, in particular
a quadratic term $\,\<\psi_*,K\psi\>_0\, $, where $\,K\, $
is a translation and (spatial) reflection invariant linear operator mapping 
$\,\bbbc^{\cX_0^{(n+1)}}\, $ to itself. If such a term were to be left
in $p_{n+1}$ it would, by the third line of \eqref{eqnOVscalingRules}, 
grow by roughly a factor of $L^2$ in each future renormalization group 
step. So we need to move (at least the local part of) this term out 
of $p_{n+1}$ and into $A_{n+1}$. By the discrete fundamental 
theorem of calculus, for any translation invariant $K$,
\begin{equation*}
\<\psi_*,K\psi\>_0\, =\, \cK\, \<\psi_*, \psi\>_0
\  +\  
\sum_{\nu=0}^3\<\psi_*,K^\nu(\partial_\nu\psi)\>_0
\end{equation*}
where $\,\cK\in\bbbc\, $ and $\,K^\nu\ ,\ \nu=0,1,2,3\,$,
are linear operators on $\,\bbbc^{\cX_0^{(n+1)}}\,$. See
\cite[Corollary \corLprelocalize]{PAR2}. By reflection invariance, $\cK$ is
real and $\sum_{\nu=1}^3\<\psi_*,K^\nu(\partial_\nu\psi)\>_0$
can be rewritten as a sum of marginal and irrelevant monomials.
See \cite[Lemma \lemLlocalize.c]{PAR2}. 

So we would like to move
$\cK \<\psi_*, \psi\>_0$ out of $p_{n+1}$ into $A_{n+1}$.
There are two factors that complicate (but not seriously) this move.
\begin{itemize}[leftmargin=*, topsep=2pt, itemsep=0pt, parsep=0pt]
\item[$\circ$]
The chemical potential term in 
$A'_{n+1}(\psi_*,\, \psi\,,
\,  \phi'_{* n+1}(\psi_*,\psi)\,,\,  \phi'_{n+1}(\psi_*,\psi))$
is 
\begin{equation*}
L^2\mu_n\< \phi'_{* n+1}(\psi_*,\psi)\, ,\, \phi'_{n+1}(\psi_*,\psi) \>_{n+1}
\end{equation*}
It is expressed in terms of $\phi'_{(*)n+1}(\psi_*,\psi)$
rather than directly in terms of $\psi_{(*)}$.
\item[$\circ$]
The prime fields $\phi'_{* n+1}(\psi_*,\psi)$,
$\phi'_{n+1}(\psi_*,\psi)$ are background fields with
chemical potential $L^2\mu_n$, not with the chemical potential $\mu_{n+1}$
that we are going to end up with (and which we do not yet know).
\end{itemize}
To deal with the first complication, we use that
$\phi'_{(*)n+1}(\psi_*,\psi)=\mathrm{B}_{(*)}\psi_{(*)}$ plus terms of
degree at least three in $(\psi_*,\psi)$ 
(see \cite[Proposition \propBGEphivepssoln.a]{BGE}). Because the 
linear operators $\mathrm{B}_{(*)}$ have left inverses 
(see \cite[Lemma \lemPOGrightinverse]{POA}
and the beginning of the proof of \cite[Lemma \lemRENpsitophi]{PAR2}),
one can show that\footnote{For reasons that will be explained
shortly, we do not actually use this fact expressed in this way.} 
$\, \cK\, \<\psi_*, \psi\>_0\, =\, 
\cK'\, \<\phi'_{* n+1}(\psi_*, \psi),\,  
\phi'_{n+1}(\psi_*, \psi)\>_{n+1}\, $
plus a power series in $\, \psi_*,\, \psi,\, \nabla\psi_*,\, 
\nabla\psi\, $ that converges on the desired domain of analyticity 
and that does not contain any relevant contributions. 
See \cite[Lemma \lemRENpsitophi]{PAR2}.
Thus
\begin{equation*}
p_{n+1}'\, =\, 
    \cK'\<\phi'_{* n+1},\, \phi'_{n+1}\>_{n+1}
    \, +\, p''_{n+1}
\end{equation*}
where  $\, p''_{n+1}\, $ has no $\<\psi_*,\psi\>_0$ term. 
Moving $\,\cK'\<\phi'_{* n+1},\,  \phi'_{n+1}\>_{n+1}\,$
from $\,p_{n+1}' $ into $\,A'_{n+1}\,$, we obtain
\begin{align*}
-A'_{n+1}\big(\psi^*,\, \psi,
\,  \phi'_{* n+1},\,  \phi'_{n+1}\big)
\, +\, 
p_{n+1}'
&=
-A''_{n+1}\big(\psi^*,\, \psi,
\,  \phi'_{* n+1},\,  \phi'_{n+1}\big)
\, +\, 
p_{n+1}''
\end{align*}
with
\begin{align*}
A''_{n+1}\big(\psi^*,\, \psi,
\,  f_*,\,  f\big)
&=
 a_{n+1} \<\psi_*\, -\,  Q_{n+1}f_*\ ,\  \psi\, -\, Q_{n+1}f\>_0 
\cr&\hskip0.5in
\,  -\, \< f_*\, ,\, (\partial_0+\Delta+(L^2\mu_n+\cK') f \>_{n+1}
\,+\, \cV'_{n+1}(f_* ,\,  f)
\end{align*}
But we are still not done --- we still have the second complication to
deal with. The prime fields $\phi'_{* n+1}(\psi_*,\psi)$,
$\phi'_{n+1}(\psi_*,\psi)$ are background fields for chemical 
potential $L^2\mu_n$, and not for chemical potential $L^2\mu_n+\cK'$. 
That is, the prime fields are critical for 
$\, f_*,f\, \mapsto\, A'_{n+1}\big(\psi_*,\,  \psi,\, f_*,\, f\big)\,$
and not for 
$\, f_*,f\, \mapsto\, A''_{n+1}\big(\psi_*,\,  \psi,\, f_*,\, f\big)\,$,
as they must be to have $A_{n+1}=A''_{n+1}$.
The way out of this is of course a (straightforward) fixed 
point  argument that yields a self consistent $\,\mu_{n+1}\approx L^2\mu_n\, $.
See \cite[Lemmas \lemRENrenormchem\ and \lemRENcRcE]{PAR2}.

\medskip
So far we have skirted the issue of bounding  
the perturbative correction $\,p_n\, $ in our main result. To measure the
size of $\,p_n\,$, we introduce a norm whose finiteness implies 
that all the kernels in its power series representation are small with  
$\,\fv_0\,$  and decay exponentially  as their arguments separate in 
$\,\cX_0^{(n)}\, $. For pedagogical simplicity pretend that $p_n$
is a function of only two fields --- $\psi$ and one derivative field 
$\psi_\nu$. It has a power series expansion
\begin{equation*}
p_n(\psi,\psi_\nu)
=
\sum_{\atop{r,s\in\bbbn_0}{r+s>0}}
\,
\sum_{ \atop{\bx\in\,{{(\cX_0^{(n)})}^r}}{\by\in\,{{(\cX_0^{(n)})}^s}} }
p_{n\,r\,s}(\bx,\by)
\, \psi(\bx)\psi_\nu(\by)
\end{equation*}
with the notations, 
$\bbbn_0=\bbbn\cup\{0\}$,
and
$
\psi(\bx)
=
\psi(x_1)\, \cdots\, \psi(x_r)
$.
Each $\,p_{n\,r\,s}(\bx,\by)\,$
is separately invariant under permutations of the components
of $\bx$ and under permutations of the components
of $\by$. The norm
of $\, p_n\, $ is
\begin{equation*}
\|p_n\|^{(n)}
=
\sum_{\atop{r,s\in\bbbn_0}{r+s>0}}
         \|p_{n\,r\,s}\|_m \ka_n^r {\ka'_n}^s
\end{equation*}
For a translation invariant kernel with four arguments, 
like the interaction kernel $\,V_0(x_1,x_2,x_3,x_4)\,$, 
$\|V_0\|_m$ is the (mass $m$) exponentially weighted 
$\, L^1$--$L^\infty$ norm of $V_0$:
\begin{equation*}
\|V_0\|_m
=
\max_{j=1,2,3,4}\, 
\sup_{x_j\in \cX_0}\,
\sum_{ \atop{x_k\in \cX_0}{k\ne j} } 
|V_0(x_1,x_2,x_3,x_4)|
\, e^{m\,\tau(x_1,x_2,x_3,x_4)}
\end{equation*}
where $\,\tau(x_1,x_2,x_3,x_4)\, $
is the minimal length of a tree graph  in $\,\cX_0\,$
that has $x_1$, $x_2$, $x_3$, $x_4$ among its vertices
and $\,m\ge 0\, $ is a fixed decay rate. (The small `coupling constant' 
$\,\fv_0=2\|V_0\|_{2m}\, $.)
The  norm $\,\|w\|_m\,$ of a kernel $\,w\, $ 
with an arbitrary number of arguments is defined in much the same way. 
For details see \cite[\S \sectINTnorms\  and Definition \defDEFkrnel]{PAR1}.

Ideally, 
$\,\|p_n\|^{(n)}\, $ would be bounded (and in fact small) uniformly in
$\,n\,$.
Unfortunately, such a bound is too naive to achieve the upper limit on $\,n\,$
stated in our main result. The reason is that, while the coefficient of
an irrelevant monomial decreases as the  scale $n$ increases, the maximum
allowed size of fields in the domain $S_n$ also increases, so the monomial 
as a whole can be relatively large. So we have chosen
\begin{itemize}[leftmargin=*, topsep=2pt, itemsep=0pt, parsep=0pt]
\item 
to move all quartic $\big(\psi_*\psi)^2$ monomials
out of $p_n$ into $A_n$, i.e. to also renormalize the interaction $V_n$, and
\item
to split $p_n$ into two parts, 
\begin{itemize}[leftmargin=*, topsep=2pt, itemsep=0pt, parsep=0pt]
\item
one, called $\cE_n(\psi_*,\psi)$,
is an analytic function whose size is measured in terms of a norm like
$\|\ \cdot\ \|^{(n)}$ and is small (and decreasing with $n$) and
\item
the other, called $\cR_n$, is a polynomial
of fixed degree, the size of whose coefficient kernels are measured 
in terms of a norm like $\|\ \cdot\ \|_m$.
\end{itemize}
\end{itemize}
The details are stated in our main result, \cite[Theorem \thmTHmaintheorem]{PAR1}.

\newpage
\appendix
\renewcommand{\theequation}{\thesection.\arabic{equation}}

\section{Seeing the Parabolic and Elliptic Regimes}
                                                 \label{sectTransition}

In this appendix we perform several model computations that contrast
the parabolic nature of the early renormalization group steps
with the elliptical nature of the late renormalization group
steps. 
We imagine that
after $n$ (block spin) renormalization group steps we have an action whose
dominant part (that we are simplifying a bit\footnote{In particular,
for pedagogical purposes, we have replaced $a_n$ by $1$ and replaced
$\cV_n$ by a local interaction.}) is 
$A_n\big(\psi_*,\psi,\phi_{*n}(\psi_*,\psi),\phi_n(\psi_*,\psi)\big)$
where
\begin{equation}\label{eqnTRaction}
\begin{split}
A_n(\psi_*,\psi,\phi_*,\phi)
&=
\<(\psi_*-Q_n\phi_*)\,,\, (\psi-Q_n\phi)\>_0
\, +\, 
\< \phi_* ,\, (-d_n\partial_0-\Delta) \phi\>_n \cr&\hskip2in
\,-\,\mu_n\< \phi_* ,\, \phi\>_n
\, +\, 
\sfrac{\rv_n}{2}\< \phi_*\phi ,\, \phi_*\phi\>_n
\end{split}
\end{equation}
Here
\begin{itemize}[leftmargin=*, topsep=2pt, itemsep=4pt, parsep=0pt]
\item[$\circ$] 
$\, \<f\, ,\, g\>_0 = \sum\limits_{x\, \in\,\cY_0 }\!\!f(x)g(x)\, $
and
$\, \<f\, ,\, g\>_n = \tveps_n\veps_n^3\sum\limits_{u\, \in\,\cY_n }f(u)g(u)\, $
are the natural real inner products on 
$\, \bbbc^{\cY_0}\, $ and
$\, \bbbc^{\cY_n}\, $, 
where the fine lattice\footnote{The fine lattice $\cY_n$ is a rescaled
version of the original lattice $\cX_0$ of \eqref{eqnOVinitialFnlInt}.} 
$\cY_n$ is  a finite periodic
box in $\tveps_n\bbbz\times\veps_n\bbbz^3$ (the lattice spacings $\tveps_n$
and $\veps_n$ are small)
and the unit lattice $\cY_0$ is  a finite periodic box in $\bbbz\times\bbbz^3$ and  
is a sublattice of $\cY_n$.

\item[$\circ$] 
$\, Q_n\, :\, \bbbc^{\cY_n}\, \rightarrow\, \bbbc^{\cY_0}\, $ 
is the linear map for which $\,(Q_nf)(x)\, $ is the average of
$\,f\, \in\, \bbbc^{\cY_n }\, $ over the square box in $\, \cY_n\, $ 
centered at $\, x\in \cY_0\, $ with sides $1$. This
box contains $\sfrac{1}{\tveps_n\veps_n^3}$ points of $\cY_n$.

\item[$\circ$]  
$\partial_0$ and $\De$ are the discrete forward time derivative
and Laplacian on $\cY_n$, respectively.

\item[$\circ$] 
 $\,\mu_n>0\, $ is the renormalized chemical potential.
 It is small in the parabolic regime and large in the elliptic regime.
$\,\rv_n>0\, $ is the renormalized coupling constant.  It is small.
$d_n>0$ is one in the parabolic regime and large in the elliptic regime.

\item[$\circ$] 
For each $\psi_*,\psi\in\bbbc^{\cY_0}$ the fields 
$\,\phi_{*n}(\psi_*,\, \psi)\, $, $\,\phi_n(\psi_*,\, \psi)\, $ on $\, \cY_n\, $
are  critical points of the functional
\begin{align*}
(\phi_*,\, \phi) 
\  
\mapsto
\ 
&A_n(\psi_*,\psi,\phi_*,\phi)
\end{align*}
They obey the background field equations
\begin{equation}\label{eqnTNbgndEqns}
\begin{alignedat}{3}
\sfrac{\de\hfill}{\de\phi_*}A_n(\psi_*,\psi,\phi_*,\phi)
&= Q_n^*(Q_n\phi-\psi) + D_n\phi + (\rv_n\phi_*\phi-\mu_n)\phi &=0\\
\sfrac{\de\hfill}{\de\phi}A_n(\psi_*,\psi,\phi_*,\phi)
&= Q_n^*(Q_n\phi_*-\psi_*) + D_n^*\phi_* + (\rv_n\phi_*\phi-\mu_n)\phi_* &=0
\end{alignedat}
\end{equation}
with $D_n=-d_n\partial_0 -\De$.

\end{itemize}

\subsection{Constant Field Background Fields}\label{sectCstFld}
To start getting a feel for the background field equations  \eqref{eqnTNbgndEqns}
we consider the case that $\psi_*$ and $\psi$ are constant fields
with $\psi_*=\psi^*$. We'll look for solutions $\phi_{(*)}$ which are also
constant fields with $\phi_*=\phi^*$. 
Since both $Q_n$ and $Q_n^*$ map the constant function $1$ to the constant
function 1, the constant field background fields obey
\begin{equation*}
\phi +\big(\rv_n|\phi|^2-\mu_n\big)\phi=\psi
\end{equation*}
This is of the form ``real number times $\phi$ equals real number times
$\psi$'' so the phase of $\phi$ and $\psi$ will be the same (modulo $\pi$).
So it suffices to consider the case that $\psi$ and $\phi$ are both real
and obey
\begin{equation*}
\phi +\big(\rv_n\phi^2-\mu_n\big)\phi=\psi
\end{equation*}
Since
\begin{equation*}
\sfrac{d\hfill}{d\phi}\big[\phi +\big(\rv_n\phi^2-\mu_n\big)\phi\big]
=1-\mu_n+3\rv_n\phi^2
\ \begin{cases}\ge 0 & \text{if $\mu_n\le 1$}
                \\[0.05in]
            >0 & \text{if  
            $\mu_n> 1$, $|\phi|>\sqrt{\sfrac{\mu_n-1}{3\rv_n}}$}
                \\[0.05in]
      <0 & \text{if  $\mu_n> 1$, $|\phi|<\sqrt{\sfrac{\mu_n-1}{3\rv_n}}$}
\end{cases}
\end{equation*}
there is always exactly one solution when $\mu_n\le 1$,
but the solution can be nonunique when $\mu_n>1$. For example, when $\mu_n>1$
and $\psi=0$ the solutions are $\phi=0$ and 
$\phi=\pm\sqrt{\sfrac{\mu_n-1}{\rv_n}}$.

\subsection{The Background Field in the Parabolic Regime}\label{sectParBgdFld}

Imagine that we wish to solve the background field equations 
\eqref{eqnTNbgndEqns} for $\phi_{(*)}$ as analytic functions of 
$\psi_{(*)}$, in the parabolic regime, when $\mu_n$ is small,
so that the minimum of the effective potential is still near the 
origin --- see \eqref{eqnOVeffectPot}. Then 
\begin{align*}
\big(Q_n^* Q_n+D_n-\mu_n\big)\phi&=Q_n^*\psi-\rv_n\phi_*\phi^2 \\
\big(Q_n^* Q_n+D^*_n-\mu_n\big)\phi_*&=Q_n^*\psi_*-\rv_n\phi_*^2\phi
\end{align*}
and, to first order in $\psi_{(*)}$,
\begin{equation}\label{eqnTRpaLinApprox}
\begin{split}
\phi&=\big(Q_n^* Q_n-\mu_n -d_n \partial_0-\De \big)^{-1}Q_n^*\psi +O\big(\psi_{(*)}^3\big)
\\
\phi_*&=\big(Q_n^* Q_n-\mu_n -d_n \partial^*_0-\De\big)^{-1}Q_n^*\psi_* 
   +O\big(\psi_{(*)}^3\big)
\end{split}
\end{equation}
We are interested in small $\psi_{(*)}$, 
so the $O\big(\psi_{(*)}^3\big)$ corrections are unimportant. 
We here see the parabolic (discrete) differential operators $d_n\partial_0^{(*)}+\De$.

\subsection{The Background Field in the Elliptic Regime}\label{sectElBgdFld}
Imagine that we again wish to solve the background  field equations 
\eqref{eqnTNbgndEqns}, but this time in the elliptic regime when $\mu_n$
is large, $\rv_n$ is small and the effective potential has a deep well,
whose minima form a circle in the complex plane of radius 
$r_n=\sqrt{\frac{\mu_n}{\rv_n}}$. We are interested in $\psi_{(*)}$
and $\phi_{(*)}$ near the minimum of the effective potential. That is, with
$\big|\psi_{(*)}\big|\,,\,\big|\phi_{(*)}\big|\approx r_n$.  We write
\begin{equation}\label{eqnTRradialTangential}
\psi = r_n e^{\Rho+i\Th}\quad
\psi_* = r_n e^{\Rho-i\Th}\quad
\phi = r_n e^{\Chi+i\Eta}\quad
\phi_* = r_n e^{\Chi-i\Eta}\quad
\end{equation}
and look for solutions when $R,\Th$ are small. 
Substitute  into \eqref{eqnTNbgndEqns} and divide by $r_n$. This gives
\begin{align*}
 D_n\big[ e^{\Chi+i\Eta}\big] 
          +Q_n^* (Q_ne^{\Chi+i\Eta}-e^{\Rho+i\Th})
          + \mu_n\big(e^{2\Chi}  - 1\big) e^{\Chi+i\Eta} 
    &=0  \\
 D_n^*\big[ e^{\Chi-i\Eta}\big] 
          +Q_n^* (Q_ne^{\Chi-i\Eta}-e^{\Rho-i\Th})
          + \mu_n\big(e^{2\Chi}  - 1\big) e^{\Chi-i\Eta} 
    &=0  
\end{align*}
Expand the exponentials, keeping only terms to first order in
$\big\{\Rho,\ \Th,\ \Chi,\ \Eta\big\}$,  to get
\begin{equation}\label{eqnTRbackGndEqnsB}
\begin{split}
D_n(\Chi+i\Eta)  +Q_n^* Q_n(\Chi+i\Eta)  + 2\mu_n\Chi &=Q_n^*(\Rho+i\Th) \\
D_n^*(\Chi-i\Eta)  +Q_n^* Q_n(\Chi-i\Eta)  + 2\mu_n\Chi &=Q_n^*(\Rho-i\Th) \\
\end{split}
\end{equation}
Now simplify, by adding together the two equations of 
\eqref{eqnTRbackGndEqnsB}
and dividing by $2$, and then subtracting the 
second equation of \eqref{eqnTRbackGndEqnsB} from the first 
and dividing by $2i$. Pretend that $\partial_0$ is a continuum partial
derivative rather than a discrete forward derivative. Then
\begin{align*}
\half(D_n+D_n^*)
&=-\sfrac{d_n}{2}(\partial_0+\partial_0^*) -\De
= -\De
\\
\sfrac{1}{2i}(D_n-D_n^*)&=\sfrac{i}{2} d_n(\partial_0-\partial_0^*)
=i\,d_n\partial_0
\end{align*}
and \eqref{eqnTRbackGndEqnsB} gives
\begin{align*}
\big[2\mu_n -\De+Q_n^* Q_n\big]\Chi
-i\,d_n \partial_0 \Eta
   &=  Q_n^* \Rho\\
i\,d_n \partial_0 \Chi
+\big[ -\De+Q_n^* Q_n\big]\Eta
  &=Q_n^*\Th
\end{align*}
or, in matrix form,
\begin{equation}\label{eqnTRbackGndEqnsH}
\square \left[\begin{matrix}\Chi\\ \Eta\end{matrix}\right] 
= Q_n^* \left[\begin{matrix}\Rho\\ \Th\end{matrix}\right]
\quad\text{or}\quad
\left[\begin{matrix}\Chi\\ \Eta\end{matrix}\right] 
= \square^{-1}  Q_n^* \left[\begin{matrix}\Rho\\ \Th\end{matrix}\right]
\end{equation}
where
\begin{equation}\label{eqnTRsquareDef}
\square=\left[\begin{matrix}
   2\mu_n -\De &
    i\,d_n \partial_0^* \\[0.05in]
    i\,d_n \partial_0 &
    -\De
   \end{matrix}\right]
+ Q_n^* Q_n 
\end{equation}
The $Q_n^* Q_n$ provides a mass which makes $\square$ boundedly 
invertible. But, the presence of this mass is a consequence of 
our having rescaled the original unit lattice down to the very fine 
lattice $\cY_n$.  To invert $\square$, ignoring the $Q_n^*Q_n$,  
we have to divide, essentially, by
\begin{align*}
\det\left[\begin{matrix}
   2\mu_n -\De &
    i\,d_n \partial_0^* \\[0.05in]
    i\,d_n \partial_0 &
    -\De
   \end{matrix}\right]
=d_n^2\big\{\partial_0^*\partial_0 + 2\sfrac{\mu_n}{d_n^2}(-\De) 
                                                                +\sfrac{1}{d_n^2}(-\De)^2\big\}
\end{align*}
\begin{itemize}
\item In the parabolic regime, $\mu_n$ is small and $d_n$ is essentially
one so that the operator in the curly brackets is approximately  
$\partial_0^*\partial_0 + (-\De)^2 $, which is parabolic.

\item In the elliptic regime, $\mu_n$ and $d_n$ are both very large
with $\sfrac{\mu_n}{d_n^2}>0$ being essentially independent
of $n$. So the operator in the curly brackets is approximately  
$\partial_0^*\partial_0 ++ 2\sfrac{\mu_n}{d_n^2}(-\De) $, which is elliptic.

\end{itemize}

\subsection{The Quadratic Approximation to the Action}\label{sectQuadAction}
For the remaining model computations, we study the quadratic
approximation to the action \eqref{sectCstFld}.

\renewcommand{\thesubsubsection}{\thesubsection.\alph{subsubsection}}
\subsubsection{Expanding Around Zero Field}
We first consider the parabolic regime as studied in \cite{PAR1,PAR2}. 
Substitute the linear approximation
to the background fields $\phi_{(*)}$ (as functions on $\psi_{(*)}$) of
\eqref{eqnTRpaLinApprox} into the action \eqref{eqnTRaction}, keeping only
terms that are of degree at most two in $\psi_{(*)}$. Writing
\begin{equation*}
S_n(\mu_n) = \big(Q_n^* Q_n-\mu_n +D_n \big)^{-1}
\end{equation*}
\eqref{eqnTRpaLinApprox} becomes
\begin{equation*}
\phi = S_n(\mu_n) Q_n^*\psi + O(\psi_{(*)}^3)\qquad
\phi_* = S_n(\mu_n)^* Q_n^*\psi_* + O(\psi_{(*)}^3)
\end{equation*}
so that
\begin{equation}\label{eqnTRparQuadApprox}
\begin{split}
&A_n
 =
\<(\psi_*-Q_n\phi_*)\,,\, (\psi-Q_n\phi)\>_0
\, +\, 
\< \phi_* ,\, (D_n-\mu_n) \phi\>_n +O(\psi_{(*)}^4) \\
&=
\<\psi_*\,,\, \psi\>_0
-\<\psi_*,\, Q_n\phi\>_0
-\<Q_n\phi_*,\, \psi\>_0
\, +\, 
\< \phi_* ,\,(Q_n^*Q_n + D_n - \mu_n) \phi\>_n 
+O(\psi_{(*)}^4) \\
&=
\<\psi_*\,,\, \psi\>_0
-\<\psi_*,\, Q_nS_n(\mu_n) Q_n^*\psi \>_0\!
-\<Q_nS_n(\mu_n)^* Q_n^*\psi_* ,\, \psi\>_0\! 
+ \< S_n(\mu_n)^* Q_n^*\psi_* ,\, Q_n^* \psi\>_n \\ &\hskip1in
+O(\psi_{(*)}^4) \\
&=
\<\psi_*,\, (\bbbone - Q_nS_n(\mu_n) Q_n^*)\,\psi\>_0
 +O(\psi_{(*)}^4) 
\end{split}
\end{equation}
We now analyse  the operator $\bbbone - Q_nS_n(\mu_n) Q_n^*$
in momentum space, in the special case that $\mu_n=0$, and see that 
it is basically a (discrete) parabolic differential operator. Set
\begin{equation}\label{eqnTRDen}
\De^{(n)}=  \big(\bbbone+ Q_n D_n^{-1} Q_n^*\big)^{-1} 
\end{equation}
Substituting in the definitions and simplifying, we see that
$S_n(0)^{-1}D_n^{-1}Q_n^*\De^{(n)}=Q_n^*$, 
so that
\begin{equation}
Q_n S_n(0)Q_n^*= Q_n D_n^{-1}Q_n^*\De^{(n)}
\end{equation}
By \cite[Remark \remPBSqnft.e]{POA}, with $\fq=1$, 
\begin{align*}
\widehat{(Q_n\phi)}(k)
&=\sum_{\ell\in \hat\cB_n} 
                  u_n(k+\ell)\,\hat\phi(k+\ell) \qquad
\widehat{(Q_n^*\psi)}(k+\ell)
   = u_n(k+\ell)\,\hat\psi(k)
\end{align*}
where
\begin{equation*}
u_n(p)
=\frac{\sin\big(\half p_0\big)}
           {\sfrac{1}{\tveps_n}\sin\big(\half \tveps_n p_0\big)}
      \prod_{\nu=1}^3 \frac{\sin\big(\half p_\nu\big)}
          {\frac{1}{\veps_n}\sin\big(\half \veps_n p_\nu\big)}
\end{equation*}
Here $k$ runs over the dual lattice of $\cY_0$ and $k+\ell$ runs
over the dual lattice of $\cY_n$. We do not need to know much
about these dual lattices, except that the dual lattice
of $\cY_0$ is a discretization of 
  $\big(\bbbr/2\pi\bbbz\big)\times \big(\bbbr^3/2\pi\bbbz^3\big)$,
the dual lattice of $\cY_n$ is a discretization of 
$\big(\bbbr/\sfrac{2\pi}{\tveps_n}\bbbz\big)
        \times \big(\bbbr^3/\sfrac{2\pi}{\veps_n}\bbbz^3\big)$,
and $\ell$ runs over
\begin{align*}
\hat\cB_n
   &=\big(2\pi\bbbz/\sfrac{2\pi}{\tveps_n}\bbbz\big)
    \!\times\!
 \big(2\pi\bbbz^3/\sfrac{2\pi}{\veps_n}\bbbz^3\big)
\end{align*}
So, by \cite[Lemmas \lemPBSunppties.b,c,
                \lemPDOhatSzeroppties.d and 
                \lemPOCDenppties.b, and
         Remark \remPOCde.a]{POA},
with $\fq=1$ and $\fQ_n$ replaced by $\bbbone$,
the operator $Q_n S_n(0) Q_n^*$ has Fourier transform
\begin{align*}
&\sum_{\ell\in\hat\cB_n}u_n(k+\ell)  \hat D_n^{-1}(k+\ell) u_n(k+\ell)
    \hat\De^{(n)}(k) \\
&\hskip0.5in = u_n(k)^2\, \hat D_n^{-1}(k)\hat\De^{(n)}(k)
   +\sum_{0\ne \ell\in\hat\cB_n}u_n(k+\ell)^2\,
     \hat D_n^{-1}(k+\ell)\hat\De^{(n)}(k) \\
&\hskip0.5in = \frac{u_n(k)^2}{
   u_n(k)^2 +\hat D_n(k) +O\big(|k|^3\big)}
   +\sum_{0\ne \ell\in\hat\cB_n}O\big(|k|^2\big)
   \prod_{\nu=0}^3 \big[ \sfrac{24}{|\ell_\nu|+\pi} \big]^2\,
   O(1)\,O(|k|) \\
&\hskip0.5in = \frac{1}{ 1  +\hat D_n(k) u_n(k)^{-2}   +O\big(|k|^3\big)}
   +\sum_{0\ne \ell\in\hat\cB_n}O\big(|k|^3\big)
   \prod_{\nu=0}^3 \big[ \sfrac{24}{|\ell_\nu|+\pi} \big]^2 \\
&\hskip0.5in = \frac{1}{
   1  +\hat D_n(k)
    +O\big(|k|^3\big)}
    +O\big(|k|^3\big) 
\end{align*}
and  $\bbbone - Q_nS_n(0) Q_n^*$ has Fourier transform
\begin{align}\label{eqnTRparaOp}
1 - \frac{1}{ 1  +\hat D_n(k)    +O\big(|k|^3\big)}
   +O\big(|k|^3\big) 
& = \hat D_n(k)  +O\big(\hat D_n(k)^2\big) +O\big(|k|^3\big) \notag\\
& = -i d_n k_0 
           +\bk^2
           + O\big(k_0^2\big)  +O\big(|k|^3\big) 
\end{align}
and so is a parabolic operator.

\subsubsection{Expanding Around the Bottom of the Effective Potential}
For all $\mu_n\ne 0$ it is appropriate to expand the action about
the bottom of the effective potential, rather than about the
origin. That is, rather than in powers of $\psi_{(*)}$. 
So we rewrite the action \eqref{eqnTRaction}
\begin{align*}
A_n
&=
\<(\psi_*-Q_n\phi_*)\,,\, (\psi-Q_n\phi)\>_0
\, +\, 
\< \phi_* ,\, (-d_n\partial_0-\Delta) \phi\>_n \\&\hskip2.5in
\, +\, 
\sfrac{\rv_n}{2}\big< {[\phi_*\phi-r_n^2]}^2,\, 1\big>_n
-\sfrac{r_n^4\rv_n}{2}\<1,1\>_n
\end{align*}
and then substitute the representations \eqref{eqnTRradialTangential}
of $\psi_{(*)}$ and $\phi_{(*)}$ in terms of radial and tangential 
fields. Note that when $\Rho=\Th=\Chi=\Eta=0$, the field magnitudes
$|\psi_{(*)}| = |\phi_{(*)}| =r_n$ and $\psi_{(*)}$ and $\phi_{(*)}$
are at the bottom of the effective potential.
Still pretending that $\partial_0$ is a continuous derivative, 
and using the notation $O[3] =  O(\Chi^3+\Rho^3+\Eta^3+\Th^3)$, we get
the following representation of the action, which is reminiscent of \eqref{eqnTRparQuadApprox}.
\begin{lemma}\label{lemTRradTranAction}
\begin{equation*}
\sfrac{1}{r_n^2} A_n 
   = \<\left[\begin{matrix}\Rho \\ \Th\end{matrix}\right]\,,\,
        \left\{ \bbbone-   Q_n \square^{-1}  Q_n^* \right\}
       \left[\begin{matrix}\Rho\\ \Th\end{matrix}\right]\>_0
         -\sfrac{r_n^2\rv_n}{2} \<1,1\>_n
         +O[3] 
\end{equation*}
\end{lemma}

\begin{proof} The three main terms in $A_n$ are
\begin{align*}
&\<(\psi_*-Q_n\phi_*)\,,\, (\psi-Q_n\phi)\>_0
=r_n^2\<(e^{\Rho-i\Th}-Q_ne^{\Chi-i\Eta})\,,\, 
           (e^{\Rho+i\Th}-Q_ne^{\Chi+i\Eta})\>_0 \\
&\hskip0.1in=r_n^2\<\Rho-i\Th-Q_n(\Chi-i\Eta)\,,\, 
           \Rho+i\Th-Q_n(\Chi+i\Eta)\>_0 
          +O[3] \\ &\hskip0.1in 
= r_n^2\<\left[\begin{matrix}\Rho \\ \Th\end{matrix}\right]\,,\,
         \left[\begin{matrix} \Rho \\ \Th \end{matrix}\right]\>_0
      -2r_n^2\<\left[\begin{matrix}\Rho \\ \Th\end{matrix}\right]\,,\,
          Q_n
         \left[\begin{matrix} \Chi \\ \Eta \end{matrix}\right]\>_0 
      +r_n^2\<\left[\begin{matrix}\Chi \\ \Eta\end{matrix}\right]\,,\,
          Q_n^* Q_n
         \left[\begin{matrix} \Chi \\ \Eta \end{matrix}\right]\>_n
         +O[3]
\end{align*}
and
\begin{align*}
\< \phi_* ,\, (-d_n\partial_0-\Delta) \phi\>_n
&=r_n^2\< e^{\Chi-i\Eta} ,\, (-d_n\partial_0-\Delta) e^{\Chi+i\Eta}\>_n \\
&=r_n^2\< \Chi-i\Eta ,\, (-d_n\partial_0-\Delta)(\Chi+i\Eta)\>_n
+O[3] \\
&=r_n^2\<\left[\begin{matrix}\Chi \\ \Eta\end{matrix}\right]\,,\,
         \left[\begin{matrix} -\Delta& i\,d_n\partial_0^*\ \\ 
                  i\,d_n\partial_0\ & -\Delta \end{matrix}\right]
         \left[\begin{matrix} \Chi \\ \Eta \end{matrix}\right]\>
   +O[3]
\end{align*}
and
\begin{align*}
\sfrac{\rv_n}{2}\< {[\phi_*\phi-r_n^2]}^2,\, 1\>_n
&=\sfrac{\rv_n r_n^4}{2} \< {[e^{2\Chi}-1]}^2,\, 1\>_n \\
&= 2r_n^2\mu_n\<\Chi,\Chi\>_n+O[3] \\
&=r_n^2\<\left[\begin{matrix}\Chi \\ \Eta\end{matrix}\right]\,,\,
         \left[\begin{matrix} 2\mu_n& 0 \\ 0 & 0  \end{matrix}\right]
         \left[\begin{matrix} \Chi \\ \Eta \end{matrix}\right]\>
         +O[3]
\end{align*}
So all together
\begin{align*}
\sfrac{1}{r_n^2} A_n 
  & = \<\left[\begin{matrix}\Rho \\ \Th\end{matrix}\right]\,,\,
         \left[\begin{matrix} \Rho \\ \Th \end{matrix}\right]\>_0
      -2\<\left[\begin{matrix}\Rho \\ \Th\end{matrix}\right]\,,\,
          Q_n
         \left[\begin{matrix} \Chi \\ \Eta \end{matrix}\right]\>_0
          \\  &\hskip0.5in
      +\<\left[\begin{matrix}\Chi \\ \Eta\end{matrix}\right]\,,\,
          \left\{\left[\begin{matrix}2\mu_n -\Delta& i\,d_n\partial_0^*\ \\ 
                  i\,d_n\partial_0\ & -\Delta \end{matrix}\right]
           +Q_n^* Q_n\right\}
         \left[\begin{matrix} \Chi \\ \Eta \end{matrix}\right]\>_n
         -\sfrac{r_n^2\rv_n}{2}\<1,1\>_n  
         +O[3]
\\[0.1in]
  & = \<\left[\begin{matrix}\Rho \\ \Th\end{matrix}\right]\,,\,
         \left[\begin{matrix} \Rho \\ \Th \end{matrix}\right]\>_0
      -2\<\left[\begin{matrix}\Rho \\ \Th\end{matrix}\right]\,,\,
          Q_n
         \left[\begin{matrix} \Chi \\ \Eta \end{matrix}\right]\>_0
      +\<\left[\begin{matrix}\Chi \\ \Eta\end{matrix}\right]\,,\,
          \square
         \left[\begin{matrix} \Chi \\ \Eta \end{matrix}\right]\>_n
          \\  &\hskip0.5in
         -\sfrac{r_n^2\rv_n}{2}\<1,1\>_n 
         +O[3]
\end{align*}
Substituting in \eqref{eqnTRbackGndEqnsH}, we have 
\begin{equation}\label{eqnTRelQadApprox}
\begin{split}
\sfrac{1}{r_n^2} A_n 
  & = \<\left[\begin{matrix}\Rho \\ \Th\end{matrix}\right]\,,\,
         \left[\begin{matrix} \Rho \\ \Th \end{matrix}\right]\>_0
      -2\<\left[\begin{matrix}\Rho \\ \Th\end{matrix}\right]\,,\,
          Q_n \square^{-1}  Q_n^* 
       \left[\begin{matrix}\Rho\\ \Th\end{matrix}\right]\>_0\\ 
         &\hskip0.3in
      +\<\square^{-1}  Q_n^* 
      \left[\begin{matrix}\Rho\\ \Th\end{matrix}\right]\,,\,
          \square
         \square^{-1}  Q_n^* \left[\begin{matrix}\Rho\\ 
                                 \Th\end{matrix}\right]\>_n
         -\sfrac{r_n^2\rv_n}{2} \<1,1\>_n
         +O[3] \\
  & = \<\left[\begin{matrix}\Rho \\ \Th\end{matrix}\right]\,,\,
        \left\{ \bbbone-   Q_n \square^{-1}  Q_n^* \right\}
       \left[\begin{matrix}\Rho\\ \Th\end{matrix}\right]\>_0
         -\sfrac{r_n^2\rv_n}{2} \<1,1\>_n
         +O[3] 
\end{split}
\end{equation}
as desired.
\end{proof}

So now we should analyse  the operator $\bbbone - Q_n \square^{-1}  Q_n^*$
in momentum space. We follow the pattern of the computation from
\eqref{eqnTRDen} through \eqref{eqnTRparaOp}.
Define
\begin{align*}
\bbbd_n&=\left[\begin{matrix}
   2\mu_n -\De &
    i d_n \partial_0^* \\[0.05in]
    i d_n \partial_0 &
    -\De
   \end{matrix}\right]
\end{align*}
and, analogously to \eqref{eqnTRDen}, 
\begin{equation*}
\tD_n= \big(\bbbone+ Q_n \bbbd_n^{-1} Q_n^* \big)^{-1}
\end{equation*}
Substituting these definitions into $\square\,\bbbd_n^{-1} Q_n^*\tD_n$
and simplifying yields $Q_n^*$ so that
\begin{equation*}
Q_n\square^{-1} Q_n^* = Q_n\bbbd_n^{-1} Q_n^*\tD_n
\end{equation*} 

The Fourier transform of $Q_n \square^{-1}  Q_n^*$ is 
\begin{equation}\label{eqnTRftQnSquareInvQnStar}
\begin{split}
&\sum_{\ell\in\hat\cB_n}u_n(k+\ell)\,
    \hat\bbbd_n^{-1}(k+\ell) u_n(k+\ell) \widehat{\tD}_n(k) \\
&\hskip0.5in = u_n(k)^2\,\hat\bbbd_n^{-1}(k)\widehat{\tD}_n(k)
    +\sum_{0\ne\ell\in\hat\cB_n}u_n(k+\ell)^2\,
    \hat\bbbd_n^{-1}(k+\ell)\widehat{\tD}_n(k) 
\end{split}
\end{equation}
where, pretending that we have continuum, rather than discrete,
differential operators, 
\begin{align*}
\hat\bbbd_n(p)
&=\left[\begin{matrix}
   2\mu_n +\bp^2 &
   d_n\,p_0 \\
   -d_n\,p_0 &
          \bp^2 
   \end{matrix}\right]
\end{align*}
and
\begin{align*}
\widehat{\tD}_n(k)
&=\Big(\bbbone+
      \sum_{\ell\in\hat\cB_n} u_n(k+\ell)^2\ 
                  \hat\bbbd_n^{-1}(k+\ell)\Big)^{-1}
\end{align*}
During the course of the upcoming computation we shall use the
following facts.
\begin{itemize}[leftmargin=*, topsep=2pt, itemsep=0pt, parsep=0pt]
\item 
The parameter $d_n\ge 1$. For small $n$ it takes
the value $1$ and for large $n$ it decays quickly approaching
$0$ as $n\rightarrow\infty$.

\item
The parameter $\mu_n > 0$. For small $n$ it is
very small and for large $n$ it is very large, with $d_n^{-2}\mu_n$
bounded uniformly in $n$. When $d_n>1$, $d_n^{-2}\mu_n$ is bounded
away from zero.

\item
 By \cite[Lemma \lemPBSunppties.b,c]{POA}.
\begin{itemize}
\item[$\circ$] $u_n(k) = 1 + O(|k|^2)$ and

  \item[$\circ$]if $\ell\ne 0$, $\big|u_n(k+\ell)\big|\le 
     \Big[\prod\limits_{\atop{0\le\nu\le 3}{\ell_\nu\ne 0}}|k_\nu| \Big]
     \prod\limits_{\nu=0}^3  \sfrac{24}{|\ell_\nu|+\pi}$.
\end{itemize}

\end{itemize}
The dominant term in \eqref{eqnTRftQnSquareInvQnStar} is
\refstepcounter{equation}\label{eqnTRdomTerm}
\begin{align}
&u_n(k)^2\,\hat\bbbd_n^{-1}(k)\widehat{\tD}_n(k)
= u_n(k)^2
    \hat\bbbd_n^{-1}(k)
    \Big\{\bbbone+
      \sum_{\ell\in\hat\cB_n} u_n(k+\ell)^2\ 
                  \hat\bbbd_n^{-1}(k+\ell)\Big\}^{-1}  \notag\\
&\hskip0.5in= \Big\{\bbbone 
    + \hat\bbbd_n(k)u_n(k)^{-2}
    +\hskip-6pt  \sum_{0\ne\ell\in\hat\cB_n} \hskip-6pt
             \sfrac{u_n(k+\ell)^2}{u_n(k)^2}\ 
                  \hat\bbbd_n^{-1}(k+\ell)
                  \hat\bbbd_n(k)\Big\}^{-1}\notag\displaybreak[0]\\
&\hskip0.5in= \bigg\{\bbbone 
    + \hat\bbbd_n(k)u_n(k)^{-2}
    +O\left(|k|^2 \left[\begin{matrix}d_n^{-2}\mu_n + |k| & d_n^{-1}|k|\\
                     d_n^{-1}\mu_n+d_n|k| &  |k| \end{matrix}\right]\right)
       \bigg\}^{-1} \notag\\
&\hskip3in\text{(by Lemma \ref{lemTRtechnicalFTbnds}.b)}
            \notag\displaybreak[0]\\[0.1in]
&\hskip0.5in=\left[\begin{matrix}
    1 +  \sfrac{2\mu_n}{u_n(k)^2} +\bk^2 + O(d_n^{-2}\mu_n|k|^2)+ O(|k|^3) &
    \sfrac{d_n k_0}{u_n(k)^2} + O(d_n^{-1}|k|^3) \notag\displaybreak[0]\\
   -d_n\,k_0 +O(d_n^{-1}\mu_n|k|^2) + O(d_n|k|^3) &
         1+ \bk^2  + O(|k|^3)
   \end{matrix}\right]^{-1}\\[0.1in]
&\hskip0.5in=\left[\begin{matrix}
    1 +  2\mu_n(1+q_1(k)) +\bk^2 + O(|k|^3)&
   d_n\,k_0 + O(d_n|k|^3) \\
   -d_n\,k_0 +O(d_n^{-1}\mu_n|k|^2) + O(d_n|k|^3) &
         1+ \bk^2  + O(|k|^3)
   \end{matrix}\right]^{-1}
\tag{\ref{eqnTRdomTerm}.a}
\end{align}
with $q_1(k)=O(|k|^2)$.
The determinant of the matrix to be inverted in the last line
of (\ref{eqnTRdomTerm}.a) is
\begin{equation}
\begin{split}
&\det\left[\begin{matrix}
    1 +  2\mu_n(1+q_1(k)) +\bk^2 + O(|k|^3)&
   d_n\,k_0 + O(d_n|k|^3) \\
   -d_n\,k_0 +O(d_n^{-1}\mu_n|k|^2) + O(d_n|k|^3) &
         1+ \bk^2  + O(|k|^3)
   \end{matrix}\right] \\
&\hskip0.5in = 1 +2\mu_n\big(1+q_1(k))\big) + 2(1+\mu_n)\bk^2
      +d_n^2k_0^2 \\ &\hskip1in
      +O\big(|k|^3\big)
      +O\big(\mu_n |k|^3\big)
      +O\big(d_n^2|k|^4\big) \\
&\hskip0.5in =d_n^2\Big\{d_n^{-2}\big[1+ 2\mu_n\big(1+q_1(k)\big)\big]
                 +k_0^2 + 2d_n^{-2}(1+\mu_n)\bk^2  +O\big(|k|^3\big)\Big\} \\
&\hskip0.5in =d_n^2\Big\{d_n^{-2} + 2d_n^{-2}\mu_n 
                           +q_2(k)  +O\big(|k|^3\big)\Big\} 
\end{split}
\tag{\ref{eqnTRdomTerm}.det}
\end{equation}
where
\begin{equation*}
q_2(k) = k_0^2 + 2d_n^{-2}\bk^2+ 2d_n^{-2}\mu_n\big(\bk^2+q_1(k)\big)
\end{equation*}

\bigskip\noindent
The tail of \eqref{eqnTRftQnSquareInvQnStar} is, by
Lemma \ref{lemTRtechnicalFTbnds}.d,
\begin{equation}
\begin{split}
&\sum_{0\ne\ell\in\hat\cB_n}u_n(k+\ell)^2 
    \hat\bbbd_n^{-1}(k+\ell)\widehat{\tD}_n(k) 
 =  O\left(|k|^2
   \left[\begin{matrix}d_n^{-4}\mu_n   + d_n^{-2}|k| & d_n^{-1}|k| \\
                       d_n^{-3}\mu_n   + d_n^{-1}|k| &   |k|
        \end{matrix}\right]  \right)
\end{split}
\tag{\ref{eqnTRdomTerm}.b}\end{equation}

Combining \eqref{eqnTRftQnSquareInvQnStar} and the three \eqref{eqnTRdomTerm}'s
we have that the Fourier transform of 
$Q_n \bbbd_n^{-1} Q_n^*{\tD}_n = Q_n\square^{-1} Q_n^*$ is
\begin{itemize}[leftmargin=*, topsep=2pt, itemsep=4pt, parsep=0pt]
\item[$\circ$]
\begin{align*}
&d_n^{-2}
    \Big\{d_n^{-2} + 2d_n^{-2}\mu_n +q_2(k) + O\big(|k|^3\big)\Big\}^{-1} 
= \sfrac{d_n^{-2}}{d_n^{-2}+2 d_n^{-2} \mu_n}
  \Big\{1- \sfrac{q_2(k)}{d_n^{-2} + 2d_n^{-2}\mu_n} +O\big(|k|^3\big)\Big\}
\end{align*}
\item[$\circ$] times
\begin{align*}
&\left[\begin{matrix}
    1+ \bk^2  + O(|k|^3)&
    -d_n\,k_0 + O(d_n|k|^3) \\
    d_n\,k_0 +O(d_n^{-1}\mu_n|k|^2) + O(d_n|k|^3) &
    1 +  2\mu_n\big(1+q_1(k)\big) +\bk^2 + O(|k|^3)
   \end{matrix}\right] 
\end{align*}
\item[$\circ$] plus
\begin{equation*}
O\left(|k|^{2}\left[\begin{matrix}d_n^{-4}\mu_n   + d_n^{-2}|k| & d_n^{-1}|k| \\
                       d_n^{-3}\mu_n   + d_n^{-1}|k| &   |k|
        \end{matrix}\right]\right)
\end{equation*}
\end{itemize}
which is
\begin{align*}
    \left[\begin{matrix} 
                1-\sfrac{2d_n^{-2}\mu_n}{d_n^{-2}+2d_n^{-2}\mu_n}
                  -q_3(k)
                          +O\big(\sfrac{\mu_n|k|^2}{d_n^4}\big)
                          +O\big(\sfrac{|k|^3}{d_n^2}\big)
                   & -\sfrac{d_n^{-1} k_0}{d_n^{-2}+2d_n^{-2}\mu_n}
                           +O\big(\sfrac{|k|^3}{d_n}\big)\\[0.05in]
                 \sfrac{d_n^{-1} k_0}{d_n^{-2}+2d_n^{-2}\mu_n}
                          +O\big(\sfrac{\mu_n|k|^2}{d_n^3}\big) 
                          +O\big(\sfrac{|k|^3}{d_n}\big) 
                & 1-q_4(k)
                   +O\big(|k|^3\big)
                    \end{matrix}\right] \\
\end{align*}
with
\begin{align*}
q_3(k) & = \sfrac{d_n^{-2}} {d_n^{-2}+2d_n^{-2}\mu_n}
           \Big\{\sfrac{q_2(k)}{d_n^{-2} + 2d_n^{-2}\mu_n}-\bk^2\Big\} \\
q_4(k) & = \sfrac{q_2(k)}{d_n^{-2} + 2d_n^{-2}\mu_n} 
    - \sfrac{2d_n^{-2}\mu_nq_1(k)+d_n^{-2}\bk^2}
               {d_n^{-2}+2 d_n^{-2} \mu_n} 
\end{align*}
So the Fourier transform of 
$\bbbone-   Q_n \square^{-1} Q_n^*$
is 
\begin{equation}\label{eqnTRelFT}
  \left[\begin{matrix} 
               \sfrac{2d_n^{-2}\mu_n}{d_n^{-2}+2d_n^{-2}\mu_n}
                  +q_3(k)
                          +O\big(\sfrac{\mu_n|k|^2}{d_n^4}\big)
                          +O\big(\sfrac{|k|^3}{d_n^2}\big)
                   & \sfrac{d_n^{-1} k_0}{d_n^{-2}+2d_n^{-2}\mu_n}
                           +O\big(\sfrac{|k|^3}{d_n}\big)\\[0.05in]
                 -\sfrac{d_n^{-1} k_0}{d_n^{-2}+2d_n^{-2}\mu_n}
                          +O\big(\sfrac{\mu_n|k|^2}{d_n^3}\big) 
                          +O\big(\sfrac{|k|^3}{d_n}\big) 
                & q_4(k)
                   +O\big(|k|^3\big)
                    \end{matrix}\right]
\end{equation}
Unraveling the definitions and simplifying gives
\begin{align*}
q_3(k) &= \sfrac{d_n^{-2}} {(d_n^{-2}+2d_n^{-2}\mu_n)^2}
           \big\{k_0^2 + d_n^{-2}\bk^2+ 2\sfrac{\mu_n}{d_n^2} q_1(k)\big\}
\\
q_4(k) &= \sfrac{1}{ d_n^{-2}+2d_n^{-2}\mu_n } \big\{
      k_0^2 + d_n^{-2}\bk^2+ 2\sfrac{\mu_n}{d_n^2}\bk^2)
\big\}
\end{align*}

When $n$ is large, that  is, deep in the ``elliptic regime'', 
the parameter $d_n\gg 1$ and $d_n^{-2}\mu_n$ is essentially constant 
and the Fourier transform of $\bbbone-   Q_n \square^{-1} Q_n^*$ is 
roughly, for small $k$
\begin{align*}
    &\left[\begin{matrix} 1
                   &  0\\[0.05in]
                 0 
                & \sfrac{k_0^2}{2d_n^{-2}\mu_n} +  \bk^2
                    \end{matrix}\right]
\end{align*}
We see an elliptic operator in the tangential direction and a
mass in the radial direction.

On the other hand, when $n$ is small, that is, early in 
the ``parabolic regime'', the parameter $d_n=1$ and $\mu_n\ll1$  
and the Fourier transform of $\bbbone-   Q_n \square^{-1} Q_n^*$
is roughly, for small $k$
\begin{align*}
    &\left[\begin{matrix} 
                   k_0^2 +  \bk^2
                   &  k_0   \\[0.05in]
                 - k_0
                & k_0^2 +   \bk^2
                    \end{matrix}\right]
\end{align*}
The eigenvalues of this matrix are
\begin{equation*}
\pm ik_0 + k_0^2 +   \bk^2 \approx \pm ik_0 +   \bk^2 
\end{equation*}
which are parabolic operators.

\subsubsection{Some Operators in Momentum Space}

We here gather together some momentum space properties of the
operators $\bbbd_n$ and $\tD_n$ that are used in the computations
leading up to \eqref{eqnTRelFT}.

\begin{lemma}\label{lemTRtechnicalFTbnds}
\ 
\begin{enumerate}[label=(\alph*), leftmargin=*]
\item 
If $p$ is bounded away from zero, then
\begin{align*}
\hat\bbbd_n^{-1}(p)
&=O\left( \left[\begin{matrix} d_n^{-2} & d_n^{-1}\\
                            d_n^{-1} &  1 \end{matrix}\right]\right)
\end{align*}

\item
If $\ell\ne 0$, then 
\begin{align*}
\hat\bbbd_n^{-1}(k+\ell) \hat\bbbd_n(k) 
&= O\left( \left[\begin{matrix}d_n^{-2}\mu_n + |k| & d_n^{-1}|k|\\
                            d_n^{-1}\mu_n+d_n|k| &  |k| \end{matrix}\right]
                            \right)
\end{align*}

\item 
\begin{align*}
\widehat{\tD}_n(k)
& = O\left(\left[\begin{matrix}d_n^{-2}\mu_n   + |k|^2 & d_n^{-1}|k|\\
                      d_n^{-1}|k| & |k|^2 \end{matrix}\right] \right)
\end{align*}

\item
If $\ell\ne 0$, then
\begin{align*}
\hat\bbbd_n^{-1}(k+\ell)\widehat{\tD}_n(k)
&= O\left( \left[\begin{matrix} d_n^{-4}\mu_n   + d_n^{-2}|k| 
                                     & d_n^{-3}|k| + d_n^{-1} |k|^2\\
                         d_n^{-3}\mu_n   + d_n^{-1}|k| 
                                & d_n^{-2}|k| + |k|^2
        \end{matrix}\right]\right)
\end{align*}

\end{enumerate}
\end{lemma}
\begin{proof}
(a) If $p$ is bounded away from zero, then
\begin{align*}
\hat\bbbd_n^{-1}(p)
&=\left[\begin{matrix}2\mu_n   +\bp^2 & d_n p_0\\
                 -d_n p_0 & \bp^2 \end{matrix}\right]^{-1}
=\sfrac{d_n^{-2}}{p_0^2 +2d_n^{-2}\mu_n \bp^2 +d_n^{-2}\bp^4}
  \left[\begin{matrix}\bp^2 & -d_n p_0\\
                 d_n p_0 & 2\mu_n   + \bp^2 \end{matrix}\right]
\\[0.1in]
&=O\left( \left[\begin{matrix} d_n^{-2} & d_n^{-1}\\
                            d_n^{-1} &  1 \end{matrix}\right]\right)
\end{align*}

\noindent (b)
If $\ell\ne 0$, then $k+\ell$ is bounded uniformly away from zero and
\begin{align*}
\hat\bbbd_n^{-1}(k+\ell) \hat\bbbd_n(k) 
& = O\left(\left[\begin{matrix} d_n^{-2} & d_n^{-1}\\
                            d_n^{-1} &  1 \end{matrix}\right]\right)
   \left[\begin{matrix}2\mu_n   +\bk^2 & d_n k_0\\
                 -d_n k_0 & \bk^2 \end{matrix}\right]
                 \\ 
&= O\left(\left[\begin{matrix}d_n^{-2}\mu_n + |k| & d_n^{-1}|k|\\
                            d_n^{-1}\mu_n+d_n|k| &  |k| \end{matrix}\right]
                            \right)
\end{align*}

\noindent (c)
Using line 4 of (\ref{eqnTRdomTerm}.a),
\begin{align*}
\widehat{\tD}_n(k)
&= u_n(k)^{-2}\,\hat\bbbd_n(k)
\big\{u_n(k)^2\,\hat\bbbd_n^{-1}(k)\widehat{\tD}_n(k)\big\} \\[0.1in]
&= \big(1+O(|k|^2\big) 
\left[\begin{matrix}2\mu_n   +\bk^2 & d_n k_0\\
                 -d_n k_0 & \bk^2 \end{matrix}\right] \\&\hskip0.5in
\left[\begin{matrix}
    1 +  \sfrac{2 \mu_n}{u_n(k)^2} +\bk^2 + O(d_n^{-2}\mu_n|k|^2)+ O(|k|^3) &
    \sfrac{d_n k_0}{u_n(k)^2} + O(d_n^{-1}|k|^3) \\
   -d_n\,k_0 +O(d_n^{-1}\mu_n|k|^2) + O(d_n|k|^3) &
         1+ \bk^2  + O(|k|^3)
   \end{matrix}\right]^{-1}
\end{align*}
Next using (\ref{eqnTRdomTerm}.det)
\begin{align*}
\widehat{\tD}_n(k)
&= \big(1+O(|k|^2\big) 
\left[\begin{matrix}2\mu_n   +\bk^2 & d_n k_0\\
                 -d_n k_0 & \bk^2 \end{matrix}\right] 
\sfrac{d_n^{-2}}{d_n^{-2}+2 d_n^{-2} \mu_n}
  \Big\{1- \sfrac{q_2(k)}{d_n^{-2} + 2d_n^{-2}\mu_n} +O\big(|k|^3\big)\Big\}\\&\hskip0.4in
\left[\begin{matrix}
    1+ \bk^2  + O(|k|^3)&
    -\sfrac{d_n k_0}{u_n(k)^2} + O(d_n^{-1}|k|^3) \\
    d_n\,k_0 +O(\sfrac{\mu_n|k|^2}{d_n}) + O(d_n|k|^3) &
    1 +  \sfrac{2 \mu_n}{u_n(k)^2} +\bk^2 
               + O(\sfrac{\mu_n|k|^2}{d_n^2})+ O(|k|^3)
   \end{matrix}\right] \\[0.1in]
&=O(1)\,d_n^{-2} \left[\begin{matrix}2\mu_n   +\bk^2 & d_n k_0\\
                 -d_n k_0 & \bk^2 \end{matrix}\right] \\&\hskip0.4in
\left[\begin{matrix}
    1+ \bk^2  + O(|k|^3)&
    -\sfrac{d_n k_0}{u_n(k)^2} + O(d_n^{-1}|k|^3) \\
    d_n\,k_0 +O(\sfrac{\mu_n|k|^2}{d_n}) + O(d_n|k|^3) &
    1 +  \sfrac{2 \mu_n}{u_n(k)^2} +\bk^2 
               + O(\sfrac{\mu_n|k|^2}{d_n^2})+ O(|k|^3)
   \end{matrix}\right] \\[0.1in]
&=O(1)\,d_n^{-2}\left[\begin{matrix}
    2\mu_n + O(d_n^2|k|^2 + \mu_n|k|^2)&
    d_n\,k_0 + O(d_n|k|^3+\sfrac{\mu_n}{d_n} |k|^3) \\
    -d_n\,k_0 +O(d_n^{-1}\mu_n|k|^4 + d_n|k|^4) &
    O(d_n^2|k|^2+\mu_n|k|^2)
   \end{matrix}\right]
\end{align*}

So 
\begin{align*}
\widehat{\tD}_n(k)
& = O\left( \left[\begin{matrix}d_n^{-2}\mu_n   + |k|^2 & d_n^{-1}|k|\\
                      d_n^{-1}|k| & |k|^2 \end{matrix}\right] \right)
\end{align*}

\noindent (d)
If $\ell\ne 0$, then
\begin{align*}
\hat\bbbd_n^{-1}(k+\ell)\widehat{\tD}_n(k)
&=O\left(\left[\begin{matrix} d_n^{-2} & d_n^{-1}\\
                            d_n^{-1} & 1 \end{matrix}\right]
   \left[\begin{matrix}d_n^{-2}\mu_n   + |k|^2 & d_n^{-1}|k|\\
                      d_n^{-1}|k| & |k|^2 \end{matrix}\right]\right) \\[0.1in]
&= O\left(\left[\begin{matrix} d_n^{-4}\mu_n   + d_n^{-2}|k| 
                                     & d_n^{-3}|k| + d_n^{-1} |k|^2\\
                         d_n^{-3}\mu_n   + d_n^{-1}|k| 
                                & d_n^{-2}|k| + |k|^2
        \end{matrix}\right]\right)
\end{align*}

\end{proof}

\newpage
\bibliographystyle{plain}
\bibliography{refs}

\begin{thebibliography}{10}

\bibitem{AGD}
A.~A. Abrikosov, L.~P. Gorkov, and I.~E. Dzyaloshinski.
\newblock {\em {Methods of Quantum Field Theory in Statistical Physics}}.
\newblock Dover Publications, 1963.

\bibitem{Bal1}
T.~Balaban.
\newblock A low temperature expansion for classical {$N$}-vector models. {I. A}
  renormalization group flow.
\newblock {\em Comm. Math. Phys.}, 167:103--154, 1995.

\bibitem{fnlint1}
T.~Balaban, J.~Feldman, H.~Kn{\"o}rrer, and E.~Trubowitz.
\newblock {A Functional Integral Representation for Many Boson Systems. I: The
  Partition Function}.
\newblock {\em Annales Henri Poincar{\'e}}, 9:1229--1273, 2008.

\bibitem{fnlint2}
T.~Balaban, J.~Feldman, H.~Kn{\"o}rrer, and E.~Trubowitz.
\newblock {A Functional Integral Representation for Many Boson Systems. II:
  Correlation Functions}.
\newblock {\em Annales Henri Poincar{\'e}}, 9:1275--1307, 2008.

\bibitem{CPC}
T.~Balaban, J.~Feldman, H.~Kn{\"o}rrer, and E.~Trubowitz.
\newblock {Power Series Representations for Complex Bosonic Effective Actions.
  I. A Small Field Renormalization Group Step}.
\newblock {\em Journal of Mathematical Physics}, 51:053305, 2010.

\bibitem{CPS}
T.~Balaban, J.~Feldman, H.~Kn{\"o}rrer, and E.~Trubowitz.
\newblock {Power Series Representations for Complex Bosonic Effective Actions.
  II. A Small Field Renormalization Group Flow}.
\newblock {\em Journal of Mathematical Physics}, 51:053306, 2010.

\bibitem{UV}
T.~Balaban, J.~Feldman, H.~Kn{\"o}rrer, and E.~Trubowitz.
\newblock {The Temporal Ultraviolet Limit for Complex Bosonic Many-body
  Models}.
\newblock {\em Annales Henri Poincar{\'e}}, 11:151--350, 2010.

\bibitem{LH}
T.~Balaban, J.~Feldman, H.~Kn{\"o}rrer, and E.~Trubowitz.
\newblock {The Temporal Ultraviolet Limit}.
\newblock In {\em {Quantum Theory from Small to Large Scales, Ecole de Physique
  des Houches, 2010}}, pages 99--170. Oxford University Press, 2012.

\bibitem{Bloch}
T.~Balaban, J.~Feldman, H.~Kn{\"o}rrer, and E.~Trubowitz.
\newblock {Bloch Theory for Periodic Block Spin Transformations}.
\newblock Preprint, 2016.

\bibitem{POA}
T.~Balaban, J.~Feldman, H.~Kn{\"o}rrer, and E.~Trubowitz.
\newblock {Operators for Parabolic Block Spin Transformations}.
\newblock Preprint, 2016.

\bibitem{SUB}
T.~Balaban, J.~Feldman, H.~Kn{\"o}rrer, and E.~Trubowitz.
\newblock {Power Series Representations for Complex Bosonic Effective Actions.
  III. Substitution and Fixed Point Equations}.
\newblock Preprint, 2016.

\bibitem{BlockSpin}
T.~Balaban, J.~Feldman, H.~Kn{\"o}rrer, and E.~Trubowitz.
\newblock {The Algebra of Block Spin Renormalization Group Transformations}.
\newblock Preprint, 2016.

\bibitem{PAR1}
T.~Balaban, J.~Feldman, H.~Kn{\"o}rrer, and E.~Trubowitz.
\newblock {The Small Field Parabolic Flow for Bosonic Many--body Models: Part 1
  --- Main Results and Algebra}.
\newblock Preprint, 2016.

\bibitem{PAR2}
T.~Balaban, J.~Feldman, H.~Kn{\"o}rrer, and E.~Trubowitz.
\newblock {The Small Field Parabolic Flow for Bosonic Many--body Models: Part 2
  --- Fluctuation Integral and Renormalization}.
\newblock Preprint, 2016.

\bibitem{PARL}
T.~Balaban, J.~Feldman, H.~Kn{\"o}rrer, and E.~Trubowitz.
\newblock {The Small Field Parabolic Flow for Bosonic Many--body Models: Part 3
  --- Nonperturbatively Small Errors}.
\newblock Preprint, 2016.

\bibitem{BGE}
T.~Balaban, J.~Feldman, H.~Kn{\"o}rrer, and E.~Trubowitz.
\newblock {The Small Field Parabolic Flow for Bosonic Many--body Models: Part 4
  --- Background and Critical Field Estimates}.
\newblock Preprint, 2016.

\bibitem{Ben}
G.~Benfatto.
\newblock {Renormalization Group Approach to Zero Temperature Bose
  Condensation}.
\newblock In {\em {Constructive Physics. Springer Lecture Notes in Physics
  446}}, pages 219--247. Springer, 1995.

\bibitem{BOG}
N.~N. Bogoliubov.
\newblock On the theory of superfluidity.
\newblock {\em J. Phys. (USSR)}, 11:23--32, 1947.

\bibitem{BryFed1}
D.~C. Brydges and P.~Federbush.
\newblock {The Cluster Expansion in Statistical Physics}.
\newblock {\em Commun. Math. Phys.}, 49:233--246, 1976.

\bibitem{BryFed2}
D.~C. Brydges and P.~Federbush.
\newblock {The Cluster Expansion for Potentials with Exponential Fall-off}.
\newblock {\em Commun. Math. Phys.}, 53:19--30, 1977.

\bibitem{CG}
S.~Cenatiempo and A.~Giuliani.
\newblock {Renormalization theory of a two dimensional Bose gas: quantum
  critical point and quasi-condensed state}.
\newblock {\em Jour. Stat. Phys.}, 157:755--829, 2014.

\bibitem{Col}
S.~Coleman.
\newblock {Secret Symmetry. An Introduction to Spontaneous Symmetry Breakdown
  and Gauge Fields.}
\newblock In {\em {Laws of Hadronic Matter}}, pages 138--215. Academic Press,
  1975.

\bibitem{FW}
A.L. Fetter and J.D. Walecka.
\newblock {\em {Quantum Theory of Many-Particle Systems}}.
\newblock McGraw-Hill, 1971.

\bibitem{GK}
K.~Gawedzki and A.~Kupiainen.
\newblock {A rigorous block spin approach to massless lattice theories}.
\newblock {\em Comm. Math. Phys.}, 77:31--64, 1980.

\bibitem{KAD}
L.P. Kadanoff.
\newblock {Scaling laws for Ising models near $T_c$}.
\newblock {\em Physics}, 2:263, 1966.

\bibitem{LSSY}
E.~H. Lieb, R.~Seiringer, J.~P. Solovej, and J.~Yngvason.
\newblock {\em {The Mathematics of the Bose Gas and its Condensation}}.
\newblock Birkh{\"a}user, 2005.

\bibitem{NO}
J.~W. Negele and H.~Orland.
\newblock {\em {Quantum Many--Particle Systems}}.
\newblock Addison--Wesley, 1988.

\bibitem{PS}
L.~Pitaevskii and S.~Stringari.
\newblock {\em {Bose--Einstein Condensation}}.
\newblock Clarendon Press, Oxford, 2003.

\bibitem{Sei}
R.~Seiringer.
\newblock {Cold quantum gases and Bose-Einstein condensation}.
\newblock In {\em {Quantum Theory from Small to Large Scales, Ecole de Physique
  des Houches, 2010}}, pages 429--466. Oxford University Press, 2012.

\bibitem{Wein}
S.~Weinberg.
\newblock {\em {The Quantum Theory of Fields, Volume II. Modern Applications}}.
\newblock Cambridge Press, 1998.

\end{thebibliography}

\end{document}